%% file: 000sampling-nodes.tex
\documentclass[sigconf,anonymous=false,nonacm=true]{acmart}

\input{001-packages}

\AtBeginDocument{%
  \providecommand\BibTeX{{%
    \normalfont B\kern-0.5em{\scshape i\kern-0.25em b}\kern-0.8em\TeX}}}

\setcopyright{none}
\renewcommand\footnotetextcopyrightpermission[1]{}
\input{002-psudecode}
\input{012-drawWhiskers}
\begin{document}

\title[Sampling Multiple Nodes in Large Networks: Beyond Random Walks]{Sampling Multiple Nodes in Large Networks:\\Beyond Random Walks}

\author{Omri Ben-Eliezer}
\affiliation{
\institution{Massachusetts Institute of Technology}
 \city{Cambridge}
 \country{Massachusetts, USA}
}
\email{omrib@mit.edu}

\author{Talya Eden}
\affiliation{
\institution{Boston University \& Massachusetts Institute of Technology, Massachusetts, USA}
 \country{}
}
\email{teden@mit.edu}

\author{Joel Oren}
\affiliation{
\institution{Bosch Center for Artificial Intelligence}
\city{Haifa}
\country{Israel}
}
\email{joel.oren@il.bosch.com}

\author{Dimitris Fotakis}
\affiliation{
\institution{National Technical University of Athens}
\city{Athens}
\country{Greece}
}
\email{fotakis@cs.ntua.gr}

\renewcommand{\shortauthors}{Ben-Eliezer, Eden, Oren and Fotakis}

\begin{abstract}
  \input{010-abstract}
\end{abstract}

\maketitle
\setdefaultleftmargin{2em}{2em}{}{}{}{} 
\setlength{\belowcaptionskip}{-12pt}

\begin{figure*}[ht!]
    \centering
    \def\imagewidth{.46\linewidth}
     \includegraphics[width=\imagewidth]{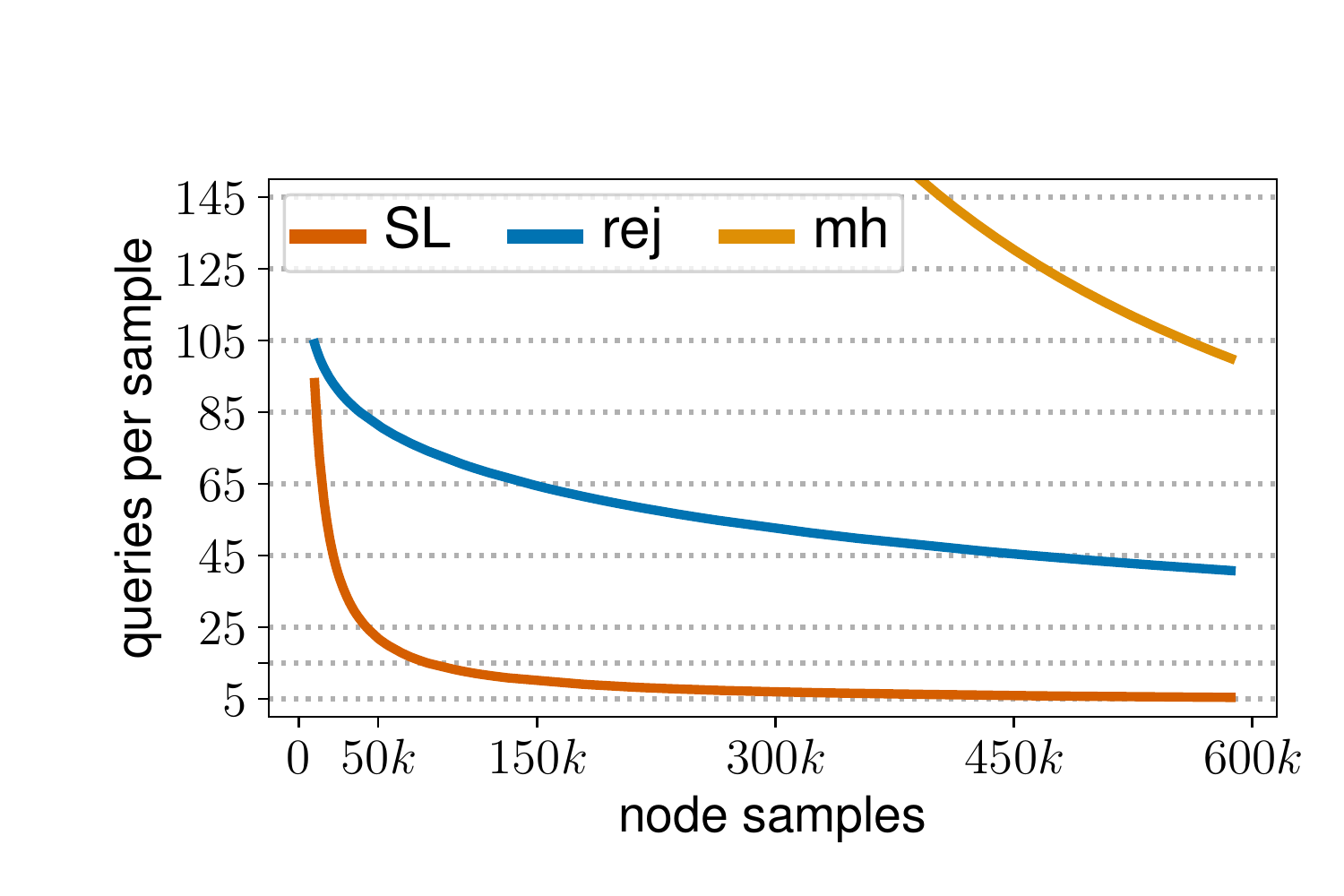}
    \hspace{-.7cm}
    \includegraphics[width=\imagewidth]{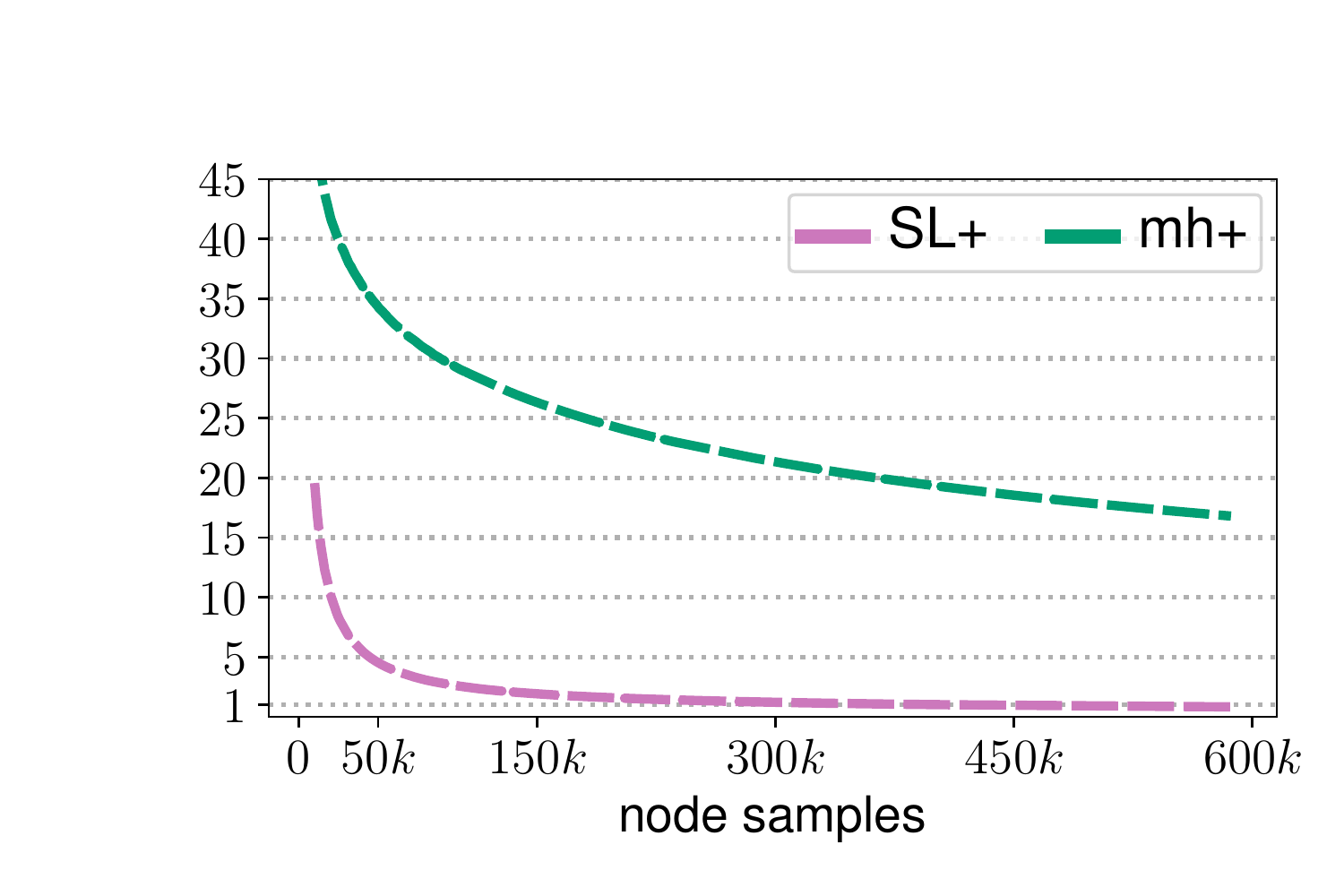}
    \hspace{.7cm}
    \vspace{-0.2cm}
\caption{Amortized query complexity per sample in
SinaWeibo, a network with 58.6M nodes and 261M edges.
Our algorithm for the standard query model, \SL, is compared to rejection sampling and Metropolis-Hastings random walks (left). The variant for the stronger query model, \SLP, is compared to the stronger variant of MH (right).}
    \label{fig:sinaweibo}
\end{figure*}

\section{Introduction}

\input{003-intro}

\subsection{Related Work}
\input{004-related}

\input{005-lower_bound}

\input{006-algorithm}

\label{sec:synopsis}

\input{013-theoretical-analysis}
\input{014-query_complexity}

\section{Empirical Results}
\label{sec:experiments}
\input{015-experiments}

\input{016-summary}

\section*{Acknowledgments}
The authors wish to thank Suman K.~Bera and C.~Seshadhri for many helpful pointers and suggestions; Flavio Chierichetti and Ravi Kumar for valuable feedback; and the anonymous reviewers, for their insightful comments.


\bibliographystyle{plain}
\bibliography{references}
\appendix
\input{020-appendix}

\end{document}

%% file: 001-packages.tex
\usepackage{amsmath, amssymb, amsthm}
\usepackage{xcolor}
\usepackage{algorithm}
\usepackage{algorithmic}
\usepackage{soul}
\usepackage{booktabs}
\usepackage{graphicx}
\usepackage[skip=4pt]{caption}
\usepackage{epsfig,multicol}
\usepackage{paralist}
\usepackage{bm}
\usepackage[font=footnotesize]{subcaption}
\usepackage{longfbox}
\usepackage{background}
\backgroundsetup{contents={}}
\usepackage{pgfplots}
\pgfdeclarelayer{background}
\pgfsetlayers{background,main}
\usepackage{calc}
\usepackage{balance}

\makeatletter
\def\blfootnote{\gdef\@thefnmark{}\@footnotetext}
\makeatother

\usepackage{hyperref}
\hypersetup{
  colorlinks   = true, 
  urlcolor     = blue, 
  linkcolor    = blue, 
  citecolor   = red 
}
\newcommand{\EX}{\mathrm{Ex}}

\newcommand{\para}[1]{\vspace{0.2cm}
\noindent\textbf{#1.}}

\newcommand{\eps}{\varepsilon}

\newcommand{\dm}{d^{-}}

 \def\notes{0} 

\ifnum\notes=1
\newcommand{\note}[2]{{\color{#1}{#2}}}

\newcommand\jnote[1]{ \note{purple}{[Joel: #1]}}
\newcommand\tnote[1]{ \note{red}{[Talya: #1]}}

\newcommand\onote[1]{ \note{blue}{[Omri: #1]}}
\newcommand\dnote[1]{ \note{purple}{[Dimitris: #1]}}
\newcommand\told[1]{ \note{black!40}{[Talya: #1]}}
\else

\newcommand{\note}[2]{}
\newcommand\jnote[1]{ }
\newcommand\tnote[1]{}
\newcommand\onote[1]{ }
\newcommand\dnote[1]{ }
\newcommand\told[1]{ }
\fi

\usepackage{xspace}

\newcommand{\SL}{\textsc{SampLayer}\xspace}
\newcommand{\SLP}{\textsc{SampLayer+}\xspace}
\newcommand{\rej}{\textsc{REJ}\xspace}
\newcommand{\mh}{\textsc{MH}\xspace}
\newcommand{\mhp}{\textsc{MH+}\xspace}

\newcommand{\mA}{\mathcal{A}}

\newcommand{\Lp}{L_{2+}}

\newcommand{\dpp}{d^{+}}

\newcommand{\brs}{rs_0}

\newcommand{\tmix}{t_{\text{mix}}}

\newtheorem{theorem}{Theorem}[section]

\usepackage[utf8]{inputenc} 
\usepackage[T1]{fontenc}    
\usepackage{hyperref}       
\usepackage{url}            
\usepackage{booktabs}       
\usepackage{amsfonts}       
\usepackage{nicefrac}       
\usepackage{microtype}      

\usepackage[export]{adjustbox}

\usepackage{tikz}
 \usepackage[usenames,dvipsnames]{pstricks}
 \usepackage{pstricks-add}
 \usepackage{epsfig}
 \usepackage{pst-grad} 
 \usepackage{pst-plot} 
 \usepackage[space]{grffile} 
 \usepackage{etoolbox} 
 \makeatletter 
 \patchcmd\Gread@eps{\@inputcheck#1}{\@inputcheck"#1"\relax}{}{}
 \makeatother
 
 \usepackage{graphicx,amsfonts,amsmath,amssymb,epsfig,color,multicol,pstricks,tikz, pgf}
\usetikzlibrary{decorations.markings}
\usetikzlibrary{patterns}
\usetikzlibrary{calc, intersections}
\usepackage{pgf,tikz}
\usetikzlibrary{calc}
\usepackage{tkz-euclide}

\definecolor{RoyalPurple}{HTML}{7851a9}
\definecolor{darkred}{HTML}{8B0000}
\definecolor{darkgreen}{HTML}{228B22}
\definecolor{teal}{HTML}{008080}
\definecolor{indigo}{HTML}{4B0082}
\definecolor{lightgray}{HTML}{D7DDDC}

%% file: 002-psudecode.tex
\newcommand{\alignInputOutput}[1]{%
	\par\addvspace{\topsep}%
	\noindent\begingroup\raggedright
	\rmfamily#1\par
	\endgroup
	\nopagebreak
  \vspace{3pt}%
\hrule
\vspace{3pt}%
}

\newcommand{\GenLZ}{\hyperref[proc:gen-L0]{\normalsize{ \color{black} \sc Generate-$ {L_0}$} (\SL)}}
\newcommand{\ALGGenLZ}{
    \alignInputOutput{
    \textbf{\textsc{Generate-$\boldsymbol{L_0}$}} \label{proc:gen-L0}\\
    \textsc{Input:} arbitrary vertex $v_0$, number $\ell_0$. \\
	\textsc{Output:} $L_0$ of size $\ell_0$, $L_1$, $\mathcal{D}$.
}
    
    \begin{compactenum}
        \item \textbf{Query} $v_0$ and let $L_0=\{v_0\}$,  $L_1=N(v_0)$.
        \item Repeat $\ell_0 - 1$ times: 
        \begin{compactenum}
        \item Pick $u \in L_1$ with maximum number of $L_0$ neighbors (break ties randomly).
        \item \textbf{Query} $u$ and remove it from $L_1$.
        \item Add $u$ to $L_0$ and add $N(u)\setminus L_0$ to $L_1$.
        \end{compactenum}
\item 
Create a data structure $\mathcal{D}$ to  sample edges between $L_0$ and $L_1$ uniformly at random. 	
\item \textbf{Return} $L_0$, $L_1$ and $\mathcal{D}$.
    \end{compactenum}
}

\newcommand{\SLtwo}{\hyperref[proc:sample-lplus]{\color{black} \sc Reach-$ {L_{\geq 2}}$}}
\newcommand{\SLtwoALG}{
	\alignInputOutput{
		\textbf{\textsc{Reach-$\boldsymbol{L_{\geq 2}}$}} \label{proc:sample-lplus}\\
		\textsc{Input:} The data structure $\mathcal{D}$. \\
		\textsc{Output:} vertex $v \in L_{\geq 2}$, its component and reach. score.
}
		\smallskip
		\begin{compactenum}
			\item  While no vertex $w$ chosen:
			\begin{compactenum} 
				\item Use $\mathcal{D}$ to sample a uniform edge  between $L_0$ and $L_1$. Let $u$ denote its $L_1$ endpoint.
				\item \textbf{Query} $u$, and if it  has neighbors in $L_{\geq 2}$, choose one of them, $w$, uniformly at random. 
			\end{compactenum} 
			\item Perform a local BFS of the component $C$ of $w$ in $G_{\geq 2}$.
	 \item Choose a vertex $v$ in $C$ uniformly at random.
	 \item Invoke \CompRS\ to compute the  reachability of $C$, $rs(C)$.
	\item \textbf{Return} $v$, and its reachability score, $rs(v)=rs(C)$.
	    \end{compactenum}
}

\newcommand{\CompRS}{\hyperref[proc:comp-rs]{\color{black} \sc Comp-Reachability}}
\newcommand{\CompRSALG}{
\alignInputOutput{
    \textbf{\textsc{Comp-Reachability}} \label{proc:comp-rs}\\
    \textsc{Input:} An (already queried) component $C$ of $G_{\geq 2}$.\\
    \textsc{Output:} The reachability score of $C$.\\
}
    \begin{compactenum}
        \item $\forall v \in C \cap L_2$:
        \begin{compactenum}
        \item 
        \textbf{Query} all $u \in N(v) \cap L_1$. For each such $u$:
            \begin{compactenum}
            \item Let  $d^-(u)=|N(u)\cap L_0|$, and  $d^+(u)=|N(u)\cap L_2|$.
            \item $\forall u \in N(v) \cap L_1$ set $rs(u)=\frac{d^-(u)}{d^+(u)}$.
            \item Set $rs(v) = \sum_{u} rs(u)$.
            \end{compactenum}
        \end{compactenum}
\item \textbf{Return} $rs(C)= \frac{1 }{ |C|} \sum_{v \in C} rs(v)$.
    \end{compactenum}
}
        
 \newcommand{\Sample}{\hyperref[proc:sample-node]{\color{black} \sc Sample}}
 \newcommand{\SampleALG}{
\alignInputOutput{
\textbf{\textsc{Sample}}
 \label{proc:sample-node}\\ 
    
    \textsc{Input:} $\bar{\ell}_{\geq 2}$ - size estimate for $L_{\geq 2}$. $rs_0$ - baseline reachability. (Both computed in the preprocessing step)\\
    \textsc{Output:} An almost uniform node in the network.\\
}
    \begin{compactenum}
    \item Choose a layer $L_0, L_1$ or ${L}_{\geq 2}$ with probability proportional to their sizes $|L_0|, |L_1|, \bar{\ell}_{\geq 2}$.
    \item If $L_0$ or $L_1$ are chosen then sample a uniform node in $L_0$ or $L_1$, respectively.
    \item If $L_{\geq 2}$ is chosen, then 
        repeatedly do:
        \begin{compactenum}
            \item Invoke \SLtwo\ and let $u$ and $rs(u)$ be the returned node and its reachability.
            \item With probability $\min(\frac{rs_0}{rs(u)}, 1)$ \textbf{return} $u$. If not returned, \textbf{repeat} loop. \label{step:rej}
        \end{compactenum}
    \end{compactenum}
}

 \newcommand{\CompRSZ}{\hyperref[proc:baseline-rs0]{\color{black} \sc Estimate-Baseline-Reachability}}
 \newcommand{\CompRSZALG}{
\alignInputOutput{
     \textbf{\textsc{Estimate-Baseline-Reachability}} \\
    \textsc{Input:} reachabilities $r_1 \leq \ldots \leq r_m$, param.~$\eps > 0$.
    \textsc{Output:} Unbiased $\eps$-quantile for the 
    reachabilities. 
}
    \begin{compactenum}
\item $\forall i \in [m]$ set $w_i = \frac{r_1}{r_i}$. Let $cw_i = \frac{\sum_{j \in [i]} w_j}{\sum_{j \in [|S|]} w_j}$.
\item \textbf{Return} minimum $r_i$ for which $cw_i \geq \eps$.
\end{compactenum}
\label{proc:baseline-rs0}
}

 \newcommand{\CompLTsize}{\hyperref[proc:comp-L2size]{\color{black} \sc Estimate-Periphery-Size}}
\newcommand{\CompLTsizeALG}{
\alignInputOutput{
\textbf{\textsc Estimate-Periphery-Size}  \\
    \textsc{Input:} $s_1$ and $s_{\geq 2}$, numbers of samples to take from $L_1$ and $L_{\geq 2}$, respectively.
    \\
    \textsc{Output:} $\overline{\ell}_{\geq 2}$ - the estimated size of $L_{\geq 2}$.
}
    \begin{compactenum}
\item \textbf{Query} $s_1$ nodes from $L_1$ uniformly at random. Denote by $S_1$ these queried nodes.
\item  $\forall v\in S_1$, compute  $d^+(v)=N(v)\cap L_{\geq 2}$.
\item Let $d^+_{avg}(S_1)=\frac{1}{|S_1|}\sum_{v \in S_1} d^+(v)$.
\item Invoke \SLtwo\ for $s_{\geq 2}$  times and let $S_2$ denote the set of returned nodes. 
\item  $\forall v \in S_2$, invoke \CompRS$(v)$ to compute the reachability score $rs(v)$.
\item $\forall v\in S_2$, compute  $d^-(v)=|N(v)\cap L_1|$.
\item Let $trs(S_2) = \sum_{v \in S_2} 1/rs(v)$ and define $d^-_{avg}(S_2)=\frac{1}{trs(S_2)}\sum_{v \in S_2} d^-(v)/rs(v)$.
\item \textbf{Return} $\overline{n}_{\geq 2}= |L_1|\cdot {{d}^+_{avg}(S_1)}\ /\  {{d}^-_{avg}(S_2)}$.
\end{compactenum}
\label{proc:comp-L2size}
}

\newcommand{\preprocessing}{

\begin{figure}[htb!]
\begin{center}
\framebox{
\begin{minipage}{.9\columnwidth}
     \CompRSZALG
\end{minipage}
}
\framebox{
\begin{minipage}{.9\columnwidth}
     \CompLTsizeALG
\end{minipage}
}
\end{center}
\caption{Missing  procedures of {\SL}.}
\label{fig:preprocess}
\end{figure}
}

\newcommand{\sampling}{
\begin{figure}[]
\raggedright
	\begin{minipage}[t]{0.954\columnwidth}
	\vspace{0pt}
		\lfbox[margin={0em,0em,0.35em,0em}]{

			\begin{subfigure}{\textwidth}
				\SampleALG
			\end{subfigure}
		}
        \lfbox[margin={0em,0em,0.35em,0em}]{
        	\begin{subfigure}{\textwidth}
        			\SLtwoALG
        		
		\end{subfigure}
        }
    \end{minipage}
	\begin{minipage}[t]{0.954\columnwidth}
		\lfbox[margin={0em,0em,0.35em,0em}]{
			\begin{subfigure}{\textwidth}
			\CompRSALG
			\end{subfigure}
		}
    	\lfbox[margin={0em,0em,0em}]{
    		\begin{subfigure}[b!]{\textwidth}
    		\BFSLtwoALG
    		\end{subfigure}
    }
        \end{minipage}

	\caption{The sampling procedures of {\SL}.
	The procedure \SLtwo\ receives as input the estimated size $\bar \ell_2$ of $L_{\geq 2}$, and the baseline reachability $rs_0$, computed in the preprocessing phase using \CompLTsize\ and \CompRSZ. It then outputs an almost uniform  node in the graph. Producing uniform nodes in $L_0$ and $L_1$ is trivial. To sample from $L_{\geq 2}$, the procedure invokes \SLtwo, which returns a node $u$ in $L_{\geq 2}$ and its reachability score. \Sample\ then performs rejection sampling with prob. proportional to $rs_0/rs(u)$. This process repeats until a node is returned. Here, we denote $L_2 = L_{\geq 2} \cap N(L_1)$ and $L_{> 2} = L_{\geq 2} \setminus L_2$.}
	
	\label{fig:sample}
\end{figure}

}

\newcommand{\CompRSALGP}{
\alignInputOutput{
    \textbf{\textsc{Comp-Reachability}} \label{proc:comp-rs-plus}\\
    \small
    \textsc{Input:} An (already visited) component $C$ of $G_{\geq 2}$.\\
    \textsc{Output:} The reachability score of $C$.\\
}
    \begin{compactenum}
    
        \item For every $v \in C$, set $rs'(v) = |N(v) \cap L_1|$.
        \item \textbf{Return} $rs(C)= \frac{1 }{ |C|} \sum_{v \in C} rs'(v)$.
    \end{compactenum}
}

%% file: 012-drawWhiskers.tex
\newcommand{\drawWhiskers}{

\definecolor{darkmidnightblue}{rgb}{0.0, 0.2, 0.4}
\definecolor{darkred}{rgb}{0.55, 0.0, 0.0}
\definecolor{darkjunglegreen}{rgb}{0.1, 0.14, 0.13}
\definecolor{darkmagenta}{rgb}{0.55, 0.0, 0.55}

\colorlet{maincolor}{blue}
\colorlet{maincolor}{maincolor!70!black}

\definecolor{l0}{HTML}{5e4376}
\definecolor{l1}{HTML}{3c4d76}
\definecolor{l2}{HTML}{2b5876}

\colorlet{dark}{l0}
\colorlet{darkBoundry}{dark!40}
\colorlet{light}{l1}
\colorlet{lightBoundry}{light!20}

\colorlet{verylight}{l2!80}
\colorlet{verylightBoundry}{verylight!5}

	\begin{tikzpicture}[xscale=.38, yscale=.33]
		
		\pgfdeclarelayer{nodelayer}
		\pgfdeclarelayer{edgelayer}
		\pgfdeclarelayer{light}
		\pgfdeclarelayer{dark}
		\pgfsetlayers{light,dark,nodelayer,edgelayer}
		
		\tikzstyle{none}=[inner sep=0,minimum size=1ex,fill=none, draw=none, shape=circle]
		
		\tikzstyle{light boundry}=[inner sep=0,minimum size=1ex,fill=none, draw=none, shape=circle]
		
		\tikzstyle{dark}=[inner sep=0,minimum size=1ex,fill=dark, draw=dark, shape=circle]
		
		\tikzstyle{circle}=[inner sep=0,minimum size=1ex,fill=verylight, draw=verylight, shape=circle]
		
		\tikzstyle{light}=[inner sep=0,minimum size=1ex,fill=light,draw=light, shape=circle]

		\begin{pgfonlayer}{nodelayer}
			\node [style=dark] (0) at (0.25, 0) {};
			\node [style=light] (1) at (-2.5, -0.25) {};
			\node [style=dark] (2) at (-1, 1) {};
			\node [style=dark] (3) at (1.25, 1) {};
			\node [style=dark] (4) at (1.5, -1.25) {};
			\node [style=dark] (5) at (-1, -1.25) {};
			\node [style=dark] (6) at (0.75, -2.25) {};
			\node [style=light] (7) at (0, 2.5) {};
			\node [style=light] (8) at (2, 3) {};
			\node [style=light] (9) at (3.5, -2) {};
			\node [style=light] (10) at (3.75, 0) {};
			\node [style=light] (11) at (-0.5, -3) {};
			\node [style=light] (12) at (2.25, -3.5) {};
			\node [style=circle] (13) at (3, 4.5) {};
			\node [style=circle] (14) at (2.25, 5.5) {};
			\node [style=circle] (15) at (4, 5.25) {};
			\node [style=circle] (16) at (5.25, 0.25) {};
			\node [style=circle] (17) at (5.75, 1.5) {};
			\node [style=circle] (18) at (7.25, 1) {};
			\node [style=circle] (19) at (6.25, -0.25) {};
			
			\node [style=circle] (20) at (4.8, -2.5) {};
			\node [style=circle] (21) at (4.8, -3.25) {};
			\node [style=circle] (22) at (6.05, -2.75) {};
			
			\node [style=circle] (23) at (-1.5, 4.25) {};
			\node [style=none] (l2) at (-2.55, 6.35) {\textcolor{verylight}{\Large{${\mathbf{L_{\geq 2}}}$}}};
			\node [style=circle] (24) at (-3.75, 3.75) {};
			\node [style=circle] (25) at (-2.5, 5.25) {};
			\node [style=circle] (26) at (-0.5, 5.5) {};
			\node [style=circle] (27) at (-3.75, 1.75) {};
			\node [style=circle] (28) at (-4.5, 0.75) {};
			\node [style=circle] (29) at (-5.25, 2) {};
			\node [style=none] (30) at (0, 2.5) {};
			\node [style=none] (31) at (1.25, 1) {};
			\node [style=none] (32) at (1.25, 1) {};
			\node [style=none] (33) at (-1, 1) {};
			
			\node [style=none] (81) at (1.5, 1.75) {};
			\node [style=none] (82) at (-1.25, 1.5) {};
			\node [style=none] (label0) at (-1.25, 2.07) {\textcolor{dark}{\Large{$\mathbf{L_0}$}}};
			\node [style=none] (83) at (-1.5, -1.5) {};
			\node [style=none] (84) at (1, -2.85) {};
			\node [style=none] (85) at (2.25, -1.25) {};
			\node [style=none] (87) at (2.25, -.68) {};
			
			\node [style=light boundry] (46) at (0, 3.25) {};
			\node [style=light boundry] (label) at (0.7, 4.1) {\textcolor{light}{\Large{${\mathbf{L_1}}$}}};
			\node [style=light boundry] (47) at (2.25, 4) {};
			\node [style=light boundry] (48) at (4.5, 0.25) {};
			\node [style=light boundry] (49) at (4.2, -2.3) {};
			\node [style=light boundry] (50) at (2.6, -4.25) {};
			\node [style=light boundry] (52) at (-1, -3.5) {};
			\node [style=light boundry] (53) at (-3.25, -0.75) {};
			\node [style=light boundry] (54) at (-2.25, 2.4) {};

			\node [style=none] (63) at (-4.5, -.25) {};
			\node [style=none] (64) at (-3, 2.2) {};
			\node [style=none] (65) at (-6, 2.6) {};
			
			\node [style=none] (67) at (1.6, 6.1) {};
			\node [style=none] (68) at (4.8, 5.6) {};
			\node [style=none] (69) at (3, 3.5) {};
			
			\node [style=none] (70) at (-1.2, 3.6) {};
			\node [style=none] (71) at (-4.6, 3.2) {};
			\node [style=none] (72) at (-2.9, 5.7) {};
			\node [style=none] (73) at (0.2, 6.1) {};
			
			\node [style=none] (74) at (4.3, -4.1) {};
			\node [style=none] (75) at (4.3, -1.85) {};
			\node [style=none] (76) at (7.2, -2.6) {};
			
			\node [style=none] (77) at (4.75, 0) {};
			\node [style=none] (78) at (5.3, 2.1) {};
			\node [style=none] (79) at (8, 1.25) {};
			\node [style=none] (80) at (6.5, -0.9) {};
			
		\end{pgfonlayer}
		\begin{pgfonlayer}{edgelayer}
			\draw[dashed] (14) to (23);
			\draw[dashed] (16.center) to (20.center);
			\draw (0) to (3);
			\draw (3) to (4);
			\draw (4) to (0);
			\draw (0) to (6);
			\draw (3) to (6);
			\draw (6) to (5);
			\draw (5) to (2);
			\draw (2) to (0);
			\draw (0) to (5);
			\draw (2) to (3);
			\draw (2) to (6);
			\draw (1) to (2);
			\draw (1) to (5);
			\draw (1) to (11);
			\draw (11) to (6);
			\draw (6) to (12);
			\draw (12) to (4);
			\draw (9) to (12);
			\draw (9) to (3);
			\draw (10) to (9);
			\draw (10) to (4);
			\draw (8) to (10);
			\draw (8) to (3);
			\draw (3) to (7);
			\draw (29) to (27);
			\draw (27) to (28);
			\draw (28) to (29);
			\draw (28) to (1);
			\draw (8) to (14);
			\draw (14) to (13);
			\draw (14) to (15);
			\draw (17) to (16);
			\draw (19) to (16);
			\draw (16) to (18);
			\draw (16) to (10);
			\draw (17) to (18);
			\draw (20) to (22);
			\draw (21) to (22);
			\draw (9) to (20);
			\draw (23) to (26);
			\draw (26) to (25);
			\draw (25) to (23);
			\draw (24) to (25);
			\draw (24) to (26);
			\draw (24) to (23);
			\draw (23) to (7);
		\end{pgfonlayer}
		\begin{pgfonlayer}{dark}
		
			\draw[rounded corners, fill=darkBoundry] 
			(81.center) 
			-- (82.center) 
			-- (83.center)
			-- (84.center)
			-- (85.center)
			-- (87.center)
			-- cycle ;
			
		\end{pgfonlayer}
		
		\begin{pgfonlayer}{light}
			\draw[rounded corners=15pt, fill=lightBoundry] 
			(46.center) 
			(47.center) 
			-- (48.center)
			-- (49.center)
			-- (50.center)
			-- (52.center)	
			-- (53.center)
			-- (54.center)
			-- cycle;
			
			\draw[rounded corners, fill=verylightBoundry] 
			(68.center)-- (69.center) -- (67.center)
			--cycle;
			
			\draw[rounded corners=10pt, fill=verylightBoundry] 
			(63.center)-- (64.center) -- (65.center)--cycle;
			
			\draw[rounded corners=10pt, fill=verylightBoundry] 
			(70.center)-- (71.center) -- (72.center)-- (73.center)--cycle;
			
			\draw[rounded corners=10pt, fill=verylightBoundry] 
			(74.center)-- (75.center) -- (76.center)--cycle;
			
			\draw[rounded corners=10pt, fill=verylightBoundry] 
			(77.center)-- (78.center) -- (79.center)-- (80.center)--cycle;
			
		\end{pgfonlayer}

	\end{tikzpicture}
}

%% file: 010-abstract.tex
Sampling random nodes is a fundamental algorithmic primitive in the analysis of massive networks, with 
many modern graph mining algorithms critically relying on it. 
We consider the task of generating a large collection of random nodes in the network assuming limited query access (where querying a node reveals its set of neighbors).
In current approaches, based on long random walks, the number of queries per sample scales linearly with the mixing time of the network, which can be prohibitive for large real-world networks. 
We propose a new method for sampling multiple nodes that bypasses the dependence in the mixing time by explicitly searching for less accessible components in the network. 
We test our approach on a variety of real-world and synthetic networks with up to tens of millions of nodes, demonstrating a query complexity improvement of up to $\times 20$ compared to the state of the art.

%% file: 003-intro.tex
Random sampling of nodes according to a prescribed distribution has been extensively employed in the analysis of modern large-scale networks for more than two decades \cite{GGR98,Klein03}. Given the massive sizes of modern networks, and the fact that they are typically accessible through 
node (or edge) queries, random node sampling offers the most natural approach, and sometimes virtually the only approach, to fast and accurate solutions for network analysis tasks. These include, for example, estimation of the order \cite{SomekhSizeEstimation}, average degree and the degree distribution \cite{DKS14,EdenJPRS18,ZKS15}, number of triangles \cite{BBCG10}, clustering coefficient \cite{SPK14}, and betweenness centrality \cite{BN19}, among many others.\blfootnote{Work partially conducted while Omri Ben-Eliezer was at Harvard University.
Talya Eden was supported by the NSF Grant CCF-1740751, Eric and Wendy Schmidt Fund for Strategic Innovation, Ben-Gurion University of the Negev and the computer science department of Boston University. Dimitris Fotakis was partially supported by NTUA Basic Research Grant (PEBE 2020) "Algorithm Design through Learning Theory: Learning‐Augmented and Data‐Driven Online Algorithms - LEADAlgo".
\\Source code: \texttt{\url{https://github.com/omribene/sampling-nodes}}.} 
Moreover, node sampling is a fundamental primitive used by many standard network algorithms for quickly exploring networks, e.g., for detection of frequent subgraph patterns~\cite{Maniacs2021} or communities~\cite{Scalable, yun2014community}, or for mitigating the effect of undesired situations, such as teleport in PageRank \cite{Gleich15}.
Hence, a significant volume of recent research is devoted to the efficiency of generating random nodes in large social and information networks; see, e.g., \cite{ChiericettiWWW2016, ChierichettiICALP18, Iwasaki2018, LiICDE2015, SamplingDirectedRibeiro2012, WalkNotWait2015, ZhouRewiring2016} and the references therein. 

\para{Problem formulation}
We consider the task of implementing a \emph{sampling oracle} that allows one
to sample multiple nodes in a network according to a prescribed distribution (say, the uniform distribution). The collection of sampled nodes should be independent and identically distributed.
Standard algorithms for random node sampling assume query access to the nodes of the network, where querying a node reveals its neighbors. As a starting point, we are given access to a single node from the network, and our goal is to be able to 
generate a (possibly large) set of 
samples from the desired distribution while performing as few queries as possible. A typical efficiency measure is the amortized \emph{query complexity} --- the total number of node queries divided by the number of sampled nodes. 
This extends the framework proposed by Chierichetti et al.~\cite{ChiericettiWWW2016, ChierichettiICALP18}, who studied the query complexity of sampling a \emph{single} node. 
For simplicity, we focus on the important case of the \emph{uniform distribution}, but our approach can be easily generalized to any natural distribution. We consider the regime where the desired number of node samples $N$ is large.

\para{Random walks and their limitations}
Most previous work on node sampling has focused on random-walk-based approaches, which naturally exploit node query access to the network. They are versatile and achieve remarkable efficiency \cite{ChiericettiWWW2016,CRS14,Iwasaki2018,LiICDE2015}. The random walk starts from a seed node and proceeds from the current node to a random neighbor, until it (almost) converges to its stationary distribution.  Then, a random node is selected according to the walk's stationary distribution, which can be appropriately modified if it differs from the desired one \cite{ChiericettiWWW2016}. The number of steps before a random walk (almost) converges to the stationary distribution is called the  \emph{mixing time}, and usually denoted by $\tmix$.

RW-based approaches are very effective in highly connected networks with good expansion properties, 
where the mixing time is logarithmic \cite{Hoory06expandergraphs}. 
In most real-world networks, however, the situation is more complicated.
It is by now a well-known phenomenon that the mixing time in many real-world social networks can be in the order of hundreds or even thousands, 
much higher than in idealized, expander-like networks   \cite{dellamico2009measurement,MohaisenMixing2010,qi2020real}.
As part of this work, we prove lower bounds for sampling multiple nodes: sampling a collection of $N$ (nearly) uniform and uncorrelated random nodes using a random walk may require 
$\Omega(N \cdot \tmix)$ queries under standard
structural assumptions. Given the multiplicative dependence in $\tmix$, this bound becomes prohibitive as $N$ grows larger.

\subsection{Our Contribution}
In light of the above discussion, 
we ask the following  question: 

\begin{center}
\textit{Can we design a highly query-efficient method for \\sampling a large number of nodes that does not depend \\multiplicatively on the network's mixing time?}
\end{center}
We answer this in the affirmative 
by presenting a novel algorithm for sampling nodes with a query complexity that is \emph{up to a factor of $20$ smaller} than that of state of the art random walk algorithms. To the best of our knowledge, this is the first node sampling method that is not based on long random walks.

\para{Lower bound for random walks}
We present an $\Omega(N\cdot \tmix)$ lower bound for sampling $N$ uniform and independent nodes from a network using naive random walks (that do not try to learn the network structure).
Our lower bound construction is a graph consisting  of a large expander-like portion and many small components connected to it by bridges, a structure that is very common among large social networks \cite{leskovec}. 
The main intuition is that sampling $N$ nodes from the network requires the walk to visit many small communities, and thus cross $\Theta(N)$ bridges, which in expectation takes $\Theta(\tmix)$ queries per bridge, and $\Theta(N \cdot \tmix)$ in total. 

\para{Bypassing the multiplicative dependency}
We present a new algorithm for sampling multiple nodes, \SL, whose query complexity does not depend multiplicatively on the mixing time. 
Our algorithm learns a structural decomposition of the network into one highly connected part and many 
small peripheral components.
We also present a stronger variant of our algorithm, \SLP, that works in a more expressive query model, where querying a node also reveals the degrees (and not just the identifiers) of its neighbors.\footnote{Such strong queries are supported, e.g., by Twitter API (\href{https://tinyurl.com/y2397aqb}{link1}, \href{https://tinyurl.com/y5aec86d}{link2}).}
Our approach is inspired by the core-periphery perspective on social networks \cite{RombachCore2014}. We start by exploring the graph with a random walk that is biased towards higher degree nodes. With the high degree nodes in hand, we build a data structure  
that partitions all non-neighbors of these nodes  
into extremely small components.
Then, we use the data structure to quickly reach these components (and subsequently, sample from them); the process involves running a BFS inside the reached component, which due to its tiny size does not require many queries.

We theoretically relate the query complexity of \SL{} to several network parameters and show that they are well-behaved in practice.
The samples generated by our algorithm are provably independent and nearly uniform.

\para{Improved empirical performance} We compare the amortized query complexity of \SL{} 
against those of the two most representative and standard random walk-based approaches for node sampling, rejection sampling (\rej) and the Metropolis-Hastings (\mh) approach; \SLP{} is compared against \mhp, an analogue of \mh in the degree-revealing model. For a complete description of \rej, \mh and \mhp with a theoretical and empirical analysis of their query complexity, see the work of Chierichetti et al.~\cite{ChiericettiWWW2016}.

We perform the comparisons on seven real world social and information networks with diverse characteristics, with the largest being SinaWeibo 
\cite{zhang2014characterizing},  which consists of more than $50$M nodes and $250$M edges.
The results presented in
Figures \ref{fig:sinaweibo} and \ref{fig:RW-comps} and in Section~\ref{subsec:comparison_exp} show that when $N$ is not extremely small,    
the query complexity of our algorithms significantly outperforms the random walk-based counterparts across the board.
This holds both under the 
standard query model (i.e., \SL vs. \rej and \mh) 
and in the stronger, degree-revealing query model (\SLP vs \mhp). 
In both models, and for all seven networks, we achieve at least $40\%$ and up to  $95\%$ reduction in the query complexity. Remarkably, as shown in Figure \ref{fig:sinaweibo}, in some cases \SL may achieve a near-optimal query complexity of as little as 5 queries per sample, even when $N$ is less than $1\%$ of the network size. In \SLP{} this is even more dramatic, essentially achieving one query per sample.

\para{Generative models}
One possible explanation to our findings might lie in the seminal work by Leskovec et al.~\cite{leskovec}. In one of the most extensive analyses of the community structure in large real-world social and information networks, they examined more than 100 large networks in various domains. They discovered that most of the classical generative models at the time did not capture well the community structure and additional various properties of social networks (e.g., size of communities, their connectivity to the graph, scaling over time, etc). The \emph{Forest Fire} generative model \cite{FF1,FF2} was developed to fill this gap, being more in-line with empirical findings. In this model, which is by now standard and well-investigated, edges are added in a way that creates small, barely connected pieces that are significantly larger and denser than random. We show that our algorithm performs very well on networks generated by this model (a consistent query complexity improvement of 30-50\%) even for tiny core sizes; see Figure \ref{fig:FF-comparison}. 
 
\begin{figure}[b]
    \centering
    \def\imagewidth{0.8\columnwidth}
    \includegraphics[width=\imagewidth]{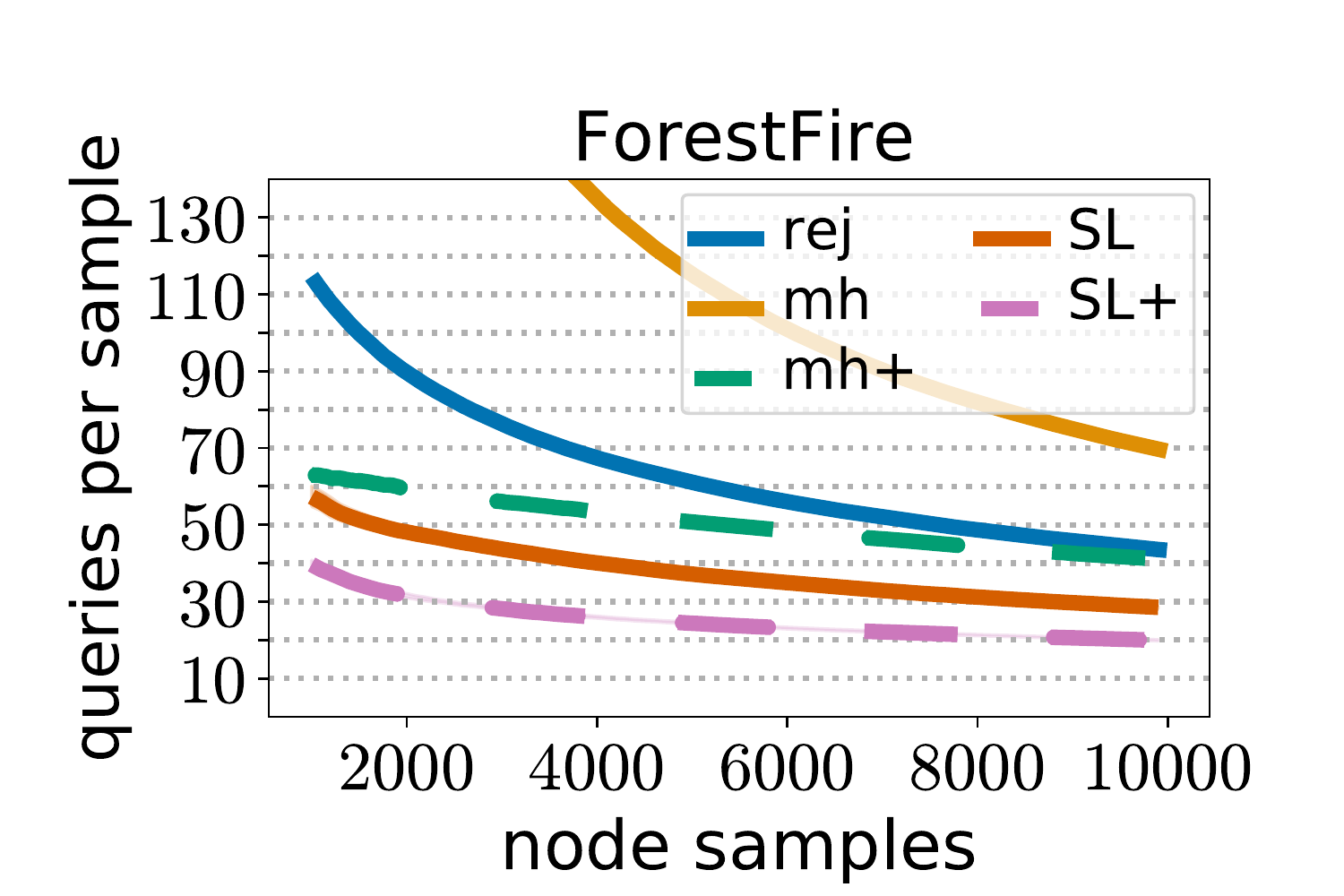}
\caption{Amortized query complexity per sample in
a Forest Fire graph with $1M$ nodes (parameters: $p_f = 0.37$, $p_b = 0.3$).}
    \label{fig:FF-comparison}
\end{figure}

%% file: 004-related.tex
Random-walk-based approaches for node sampling have been  studied extensively in the last decade. The aforementioned work of Chierichetti et al. \cite{ChiericettiWWW2016} is the closest to ours, studying the query complexity of such approaches. The analysis of Iwasaki and Shudo \cite{Iwasaki2018} also focuses on the average query complexity of random walks. 

Using random node sampling to determine the properties of large-scale networks goes back to the seminal work of Leskovec and Faloutsos \cite{LeskovecSampling2006}. Since then, the performance of node sampling via random walks has been widely considered in the context of network parameter estimation. For example, Katzir et al.~\cite{SomekhSizeEstimation,KH15} use random walks to estimate the network order and the clustering coefficient based on sampling and collision counting. Cooper et al.~\cite{CRS14} present a general random-walk-based framework for estimating various network parameters; whereas Ribeiro et al.~\cite{SamplingDirectedRibeiro2012} extend these approaches to directed networks. Ribeiro and Towsley~\cite{RT10} use multidimensional random walks. Eden et al.~\cite{ER18,ERS18,EMR21} and T\v{e}tek and Thorup~\cite{TT} study the query complexity of generating uniform edges given access to uniform nodes. Bera and Seshadhri \cite{BeraSeshadhri20} devise an  accurate sublinear triangle counting algorithm that queries only a small fraction of the graph edges.
Several other examples of network estimation works can be found in the literature \cite{WalkingInFacebook, JinAlbatross2011, LiICDE2015, WalkNotWait2015, ZhouRewiring2016, KSY20}.

In many of these works, improved query complexity is achieved by relaxing the requirement for independent samples. Understanding how  dependencies between sampled nodes affect the outcome, however, inherently requires a complicated analysis tailored specifically for the network-parameter at question. Such an analysis is not required for nodes generated by our approach, which are provably independent.
From a technical viewpoint, our adaptive exploration of the network's isolated components bears some similarity to node sampling via deterministic exploration \cite{Salamanos2017} and to \emph{node probing} approaches (e.g., \cite{BDD14,SEGP15,SEGP17,LSBE18}) for network completion \cite{HX09,KL11}. Given access to an incomplete copy of the actual network, Soundarajan et al.~\cite{SEGP15,SEGP17} and LaRock et al.~\cite{LSBE18} discover the unobserved part 
via adaptive network exploration; see also the survey by Eliassi-Rad et al.~\cite{Eliassi19}. 

Utilizing core-periphery characteristics of networks for algorithmic purposes has received surprisingly little attention. The most relevant work is by Benson and Kleinberg \cite{BK19}, on link prediction.

%% file: 005-lower_bound.tex
\section{Lower Bound For Random Walks}
\label{sec:lower_bound}

In this section, we quickly present the $\Omega(N \cdot \tmix)$ lower bound on the number of queries required to sample $N$ (nearly) independent and uniformly distributed nodes using random walks. See proof sketch in Section~\ref{sec:missing_lb}. For clarity, we focus our analysis on the most standard random walk, which at any given time proceeds from the current node to one of its neighbors, uniformly at random; such a random walk is used in rejection sampling (\rej). Similar lower bounds hold for other random walk variants that do not learn structural characteristics of the network, including \mh and \mhp.  

\begin{theorem}\label{thm:lb}
For any $n$ and $\log n\ll t \ll n^{\Theta(1)}$, there exists a graph $G$ on $n$ vertices 
with mixing time $\tmix=\Theta(t)$, that satisfies the following: for any $N \leq n^{\Theta(1)}
$, any sampling algorithm based on uniform random walks that outputs a (nearly) uniform collection of $N$ nodes must perform $\Omega(N\cdot \tmix)$ queries.
\end{theorem}

%% file: 006-algorithm.tex
\begin{figure}
\centering
  \def\imagelen{0.49
  \linewidth}
    \raisebox{0pt}{\includegraphics[width=\imagelen, height=2.9cm]{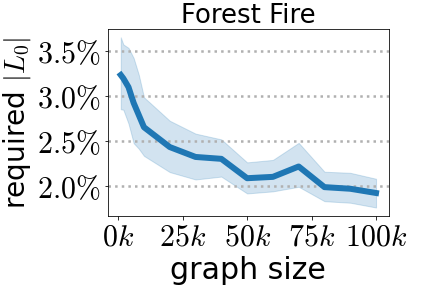}}
  \includegraphics[width=\imagelen, height=2.9cm]{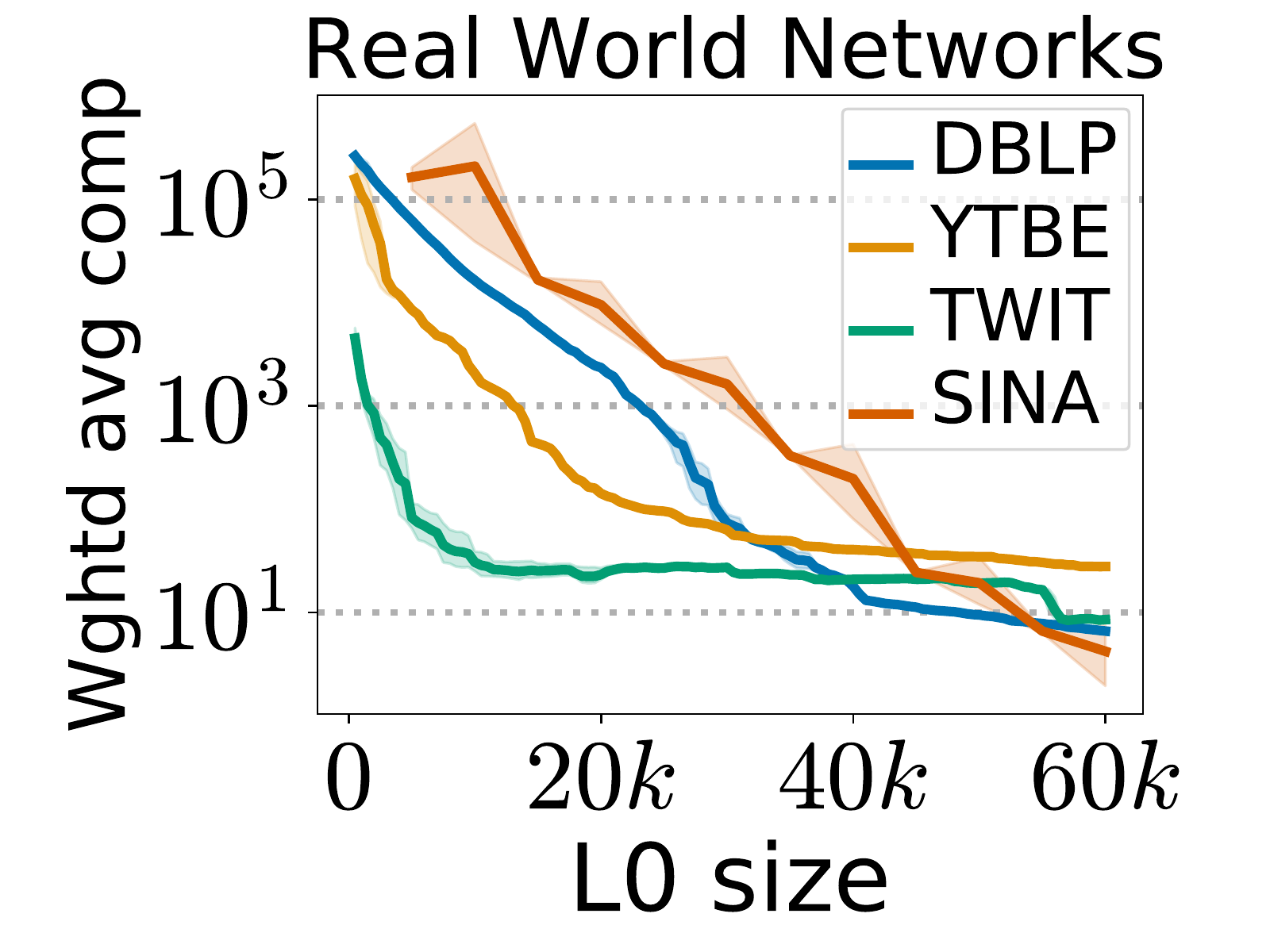}
  \caption{Very small $L_0$ suffices to shatter $L_{\geq 2}$-components in a Forest Fire graph (left); convergence of $L_{\geq 2}$ component sizes as $|L_0|$ grows, in four real-world networks (right).}
  \label{fig:L0-L2-exps}
\end{figure}

\section{Algorithm}\label{sec:ub}

In order to bypass the multiplicative dependence in the mixing time, one needs to exploit structural characteristics of social networks in some way. One natural property is that the degree-distribution is top-heavy; furthermore, 
a large fraction of nodes in the network are well-connected to the high-degree nodes, whereas the remaining nodes decompose into small, weakly connected components (e.g., \cite{leskovec, RombachCore2014}). In Figure \ref{fig:L0-L2-exps} we demonstrate this phenomenon in a strong quantitative form. 
Suppose  that we are able to access a collection of, say, the top $1\%$ highest degree nodes in the network, and call these $L_0$. Denote their neighbors by $L_1$ and the rest of the network by $L_{\geq 2}$. Is it the case that a small $L_0$ size suffices for $L_{\geq 2}$ to decompose into tiny isolated components?

Our experiments indicate that the answer is positive, even if one instead generates $L_0$ greedily with our query access, 
starting from an arbitrary seed vertex. The details are given in the experimental section, but briefly, Figure \ref{fig:L0-L2-exps} demonstrates that for both the Forest Fire model with standard parameters and for various real-world social networks with diverse characteristics, a very small $L_0$ size (ranging between $0.1\%$ and $10\%$ of the graph size, and in most cases about $1-2\%$) suffices for $L_{\geq 2}$ to decompose very effectively.

These results suggest a new approach to quickly reach nodes in the network. In a preprocessing phase, greedily capture $L_0$ as above, which decomposes the rest of the network into $L_1$ and $L_{\geq 2}$. $L_1$-nodes are easy to reach; $L_{\geq 2}$-nodes are reachable by attempting to visit a component from $L_{\geq 2}$ through a walk of length $2$ from $L_{0}$, and then fully exploring the component via a BFS. Our algorithm, \SL, is based upon this idea, as well as running  size  and reachability estimations  to ensure the generated samples are close to uniform.

\subsection{Algorithm Description}
\label{sec:alg_desc}

\begin{figure}[]
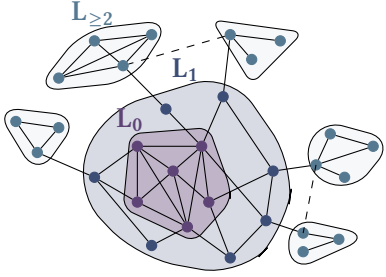

    \centering
    \drawWhiskers
    \caption{An illustration of a typical network and its layering by our algorithm, \SL. The $L_{\geq 2}$-layer components intuitively correspond to  small communities  that are weakly connected to the rest of the network.
    \vspace{0.1cm}
    }
    \label{fig:whiskers_net}
\end{figure}
We next describe our algorithm, \SL, in detail. 
The algorithm runs in two phases: a structural learning phase and a sampling phase.
In the first phase, the algorithm constructs a data structure providing fast access to nodes that are either very highly-connected (we call these nodes the $L_0$-layer) or neighbors thereof (the $L_1$-layer). 
This exploits the well-known fact that in large social and information networks, typically a large fraction of the nodes are connected to a highly influential core \cite{RombachCore2014}.
The node sampling itself takes place in the second phase, which uses the data structure to either sample from the core layers $L_0\cup L_1$, or to explicitly cross bridges that lead to the small, less connected parts. These are edges from $L_1$ to nodes outside $L_0 \cup L_1$.  We refer to  these nodes as the $L_{\geq 2}$ layer.
Once it reaches such a small component, the algorithm fully explores the component and uniformly samples a node from within it. Finally, our algorithm uses rejection sampling to ensure that (almost) all nodes are returned with equal probability. 

\para{Structural decomposition phase}  Starting from an arbitrary node, we aim to capture the  highest-degree nodes in the network. This is done by our procedure \textsc{Generate-$L_0$} below.
We add these nodes greedily, one by one, where intuitively, in every step the newly added node is the one we perceive (according to the information currently available) as the highest-degree one.
We refer to this initial collection of high-degree nodes as the base layer, $L_0$.\footnote{We note that one can modify \textsc{Generate-$L_0$} by first conducting a random walk using part of the $L_0$-construction budget, and only then continuing with the above process. While the added randomness could theoretically help escaping situations where the initial node is problematic in some way or there are multiple cores in the graph, in all networks that we tested adding such a random walk did not improve the quality of $L_0$; in fact, the existing algorithm captured essentially all nodes with very high degrees.}

\begin{figure}
\begin{center}
\framebox{
\begin{minipage}{.9\columnwidth}
     \ALGGenLZ
\end{minipage}
}
\end{center}
\label{fig:pseudoGenL0}
\end{figure}

\input{017-generate_L0}

\input{011-figure-reach-SL}
\begin{figure}
\begin{center}
\vspace{0.7cm}
\framebox{
\begin{minipage}{.9\columnwidth}
     \SampleALG
\end{minipage}
}
\framebox{
\begin{minipage}{.9\columnwidth}
     \SLtwoALG
\end{minipage}
}
\framebox{
\begin{minipage}{.9\columnwidth}
    \CompRSALG
\end{minipage}
}
\end{center}
\caption{The sampling procedures.} 
    \label{fig:sample}
\end{figure}

The next layer, $L_1$, is the set of neighbors of $L_0$ that are not already in $L_0$, i.e. $L_1 = \bigcup_{v \in L_0} N(v)\setminus{L_0}$, where $N(v)$ denotes the set of neighbors of node $v$. Intuitively,  the union of these two  layers  captures the well-connected or ``expanding'' part of the network. 
The neighbors of $L_1$ are  denoted $L_2$ and the multi-layer consisting of all  other nodes in the network is denoted by $L_{> 2}$, where we also set $L_{\geq 2}=L_2\cup L_{>2}$.
See 
Figure~\ref{fig:whiskers_net} for a visualization of the layers 
and Figure~\ref{fig:gen-L0} 
for a visualization of the structural decomposition phase of \SL\ and its variant \SLP.

Denote  by $G_{\geq 2}$ the subgraph whose node set is $L_{\geq 2}$, and whose edge set includes all edges between $L_2$ and $L_{>2}$ and all edges between nodes in $L_{>2}$. 
Crucially, the size $\ell_0$  of the generated $L_0$ should be sufficiently large so that the subgraph $G_{\geq 2}$ will ``break'' into many small connected components. (Note that we intentionally ``ignore'' edges between vertices that lie strictly in $L_2$, to make these components as small as possible.) In Section~\ref{sec:other_exps}, we explore the typical size of $G_{\geq 2}$-components as a function of $\ell_0$, and discuss how to determine the ``correct'' $\ell_0$ value for the network at hand.

To complete this phase, we learn various parameters of the layering that are crucial for the sampling phase, including accurate approximations of the size of $L_{\geq 2}$  and the typical reachability of nodes in it. This is done using the procedures \CompLTsize{} and \CompRSZ, given in 
Section~\ref{sec:preprocess_missing}.
\onote{Supplementary~\ref{sec:preprocess_missing} - in general we should look for all references to the appendix and say they are in supplementary material}
Specifically, estimating the size of $L_{\geq 2}$ is done by  considering the bipartite graph  with $L_1$ on one side and $L_{\geq 2}$ on the other. 
By sampling $s_1$ nodes from $L_1$ and $s_{\geq 2}$ nodes from $L_{\geq 2}$ (using the procedure \SLtwo{}), we can estimate the average degrees of the nodes of each side of the bipartite graph, from which we can estimate the size of $L_{\geq 2}$. The reachability distribution is approximated by calculating the reachabilities of the $s_{\geq 2}$ samples from $L_{\geq 2}$. This procedure receives as an input a parameter $\eps$, and returns a ``baseline reachability'' which is approximately the $\eps$-percentile of $L_{\geq 2}$-nodes in terms of reachability.
In Section \ref{subsec:comparison_exp} we discuss how to practically choose the parameters $s_1$, $s_{\geq 2}$, and $\eps$. 

\para{Sampling phase} 
Sampling from the core layers $L_0$ and $L_1$ is trivial;
the challenge is to sample efficiently from $L_{\geq 2}$. 
Taking advantage of the
layering, we sample random nodes in $L_{\geq 2}$ by combining walks of length $2$ that start in $L_0$ and reach $L_{\geq 2}$, with a local BFS step that explores
 and returns a uniformly selected node  in the reached $L_{\geq 2}$ component.
The above process generates biased samples, as the vertices in different components have different probabilities to be reached in the initial 2-step walk.
Hence, the final step in the sampling procedure is a rejection step, whose role is to unbias the distribution. To this end, we compute a suitable \emph{reachability score}, $rs(v)$  for every reached vertex. 
We then perform a rejection step, where the acceptance probability is inversely proportional to the reachability score of the chosen node.
See Figure~\ref{fig:sample} for the pseudo-code and  Figure~\ref{fig:reach-SL} for an illustration of the sampling process in \SL.

\para{Non-uniform distributions}
For simplicity, our algorithm is presented for node generation according to the uniform distribution. We note that it can be adapted to generate other desirable distributions. 
For example, to conduct $\ell_p$-sampling, the size estimation procedure should be replaced by a procedure that estimates the sum $\sum_{v \in L_{\geq 2}} (d(v))^p$ (and the corresponding sum for $L_{1}$), and the reachability distribution estimation should be adjusted accordingly.

%% file: 017-generate_L0.tex
 \newcommand{\drawGenerateLzero}[3]{

\tikzstyle{client}=[draw=none,circle,minimum size=2.5ex, inner sep=0, fill,]

\node[client, fill=#1] (A) at (0,0) {};
\node[client, fill=#2] (B) at (1,0) {};
\node[client, fill=#2] (C) at (-0.4,0.8) {};
\node[client, fill=#2] (D) at (-0.7,-0.6) {};
\node[client, fill=#2] (P) at (1.4,0.8) {};
\node[client, fill=#3] (E) at (1.7, 1.9) {};
\node[client, fill=#3] (F) at (0.3, 1.6) {};
\node[client, fill=#3] (G) at (-1.8,-1.9) {};
\node[client, fill=#3] (H) at (-1.2,0.9) {};
\node[client, fill=#3] (I) at (1.7,-1.7) {};
\node[client, fill=#3] (J) at (-1.8,1.7) {};
\node[client, fill=#3] (K) at (-0.9,1.9) {};
\node[client, fill=#3] (L) at (-1.9,0.4) {};
\node[client, fill=#3] (M) at (-0.5,-1.9) {};
\node[client, fill=#3] (N) at (0,-1.2) {};
\node[client, fill=#3] (O) at (0.4,-1.6) {};
\node[client, fill=#3] (Q) at (1.2,-0.9) {};
\node[client, fill=#3] (R) at (-1.6,-0.5) {};

\tikzstyle{edgeStyle}=[#2, very thick]
\draw[edgeStyle] (A) -- (B);
\draw[edgeStyle] (A) -- (C);
\draw[edgeStyle] (A) -- (D);
\draw[edgeStyle] (A) -- (P);

\tikzstyle{edgeStyle}=[#3, very thick]
\draw[edgeStyle] (B) -- (P);
\draw[edgeStyle] (B) -- (Q);
\draw[edgeStyle] (C) -- (D);
\draw[edgeStyle] (C) -- (H);
\draw[edgeStyle] (C) -- (F);
\draw[edgeStyle] (C) -- (K);
\draw[edgeStyle] (C) -- (R);
\draw[edgeStyle] (D) -- (R);
\draw[edgeStyle] (D) -- (N);
\draw[edgeStyle] (E) -- (F);
\draw[edgeStyle] (E) -- (P);
\draw[edgeStyle] (G) -- (M);
\draw[edgeStyle] (H) -- (J);
\draw[edgeStyle] (H) -- (K);
\draw[edgeStyle] (H) -- (L);
\draw[edgeStyle] (H) -- (R);
\draw[edgeStyle] (I) -- (Q);
\draw[edgeStyle] (M) -- (N);
\draw[edgeStyle] (N) -- (O);
\draw[edgeStyle] (O) -- (Q);

}

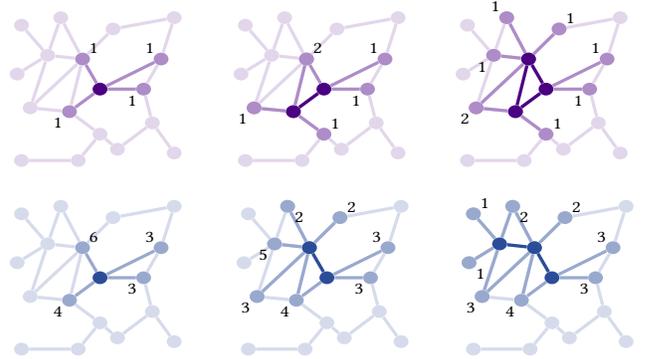
\begin{figure}

\def\Figonelen{0.14\textwidth}

  \begin{minipage}[b]{\Figonelen}
  
    \begin{tikzpicture}
    \begin{scope}[xscale=0.58, yscale=.5, transform shape]
    \drawGenerateLzero{indigo}{indigo!45}{indigo!16};
    \node[scale=1.4, below left] at (B) {1};
    \node[scale=1.4, above right] at (C) {1};
    \node[scale=1.4, below left] at (D) {1};
    \node[scale=1.4, above left] at (P) {1};
    \end{scope}
    \end{tikzpicture}    
  \end{minipage}
  \hfill
  \begin{minipage}[b]{\Figonelen}

    \begin{tikzpicture}
    \begin{scope}[xscale=0.58, yscale=.5, transform shape]
    \drawGenerateLzero{indigo}{indigo!45}{indigo!16};
    \node[client, fill=indigo] at (D) {};
    \node[client, fill=indigo!45] at (R) {};
    \node[client, fill=indigo!45] at (C) {};
    \node[client, fill=indigo!45] at (N) {};
    \tikzstyle{edgeStyle}=[indigo, very thick]
    \draw[edgeStyle] (A) -- (D);
    \tikzstyle{edgeStyle}=[indigo!45, very thick]
    \draw[edgeStyle] (D) -- (N);
    \draw[edgeStyle] (D) -- (R);
    \draw[edgeStyle] (D) -- (C);
    \node[scale=1.4, below left] at (B) {1};
    \node[scale=1.4, above right] at (C) {2};
    \node[scale=1.4, above left] at (P) {1};
    \node[scale=1.4, above right] at (N) {1};
    \node[scale=1.4, below left] at (R) {1};
    \end{scope}
    \end{tikzpicture}
  \end{minipage}
  \hfill
  \begin{minipage}[b]{\Figonelen}

    \begin{tikzpicture}
    \begin{scope}[xscale=0.58, yscale=.5, transform shape]
    \drawGenerateLzero{indigo}{indigo!45}{indigo!16};
    \node[client, fill=indigo] at (D) {};
    \node[client, fill=indigo] at (C) {};
    \node[client, fill=indigo!45] at (R) {};
    \node[client, fill=indigo!45] at (N) {};
    \node[client, fill=indigo!45] at (H) {};
    \node[client, fill=indigo!45] at (F) {};
    \node[client, fill=indigo!45] at (K) {};
    \tikzstyle{edgeStyle}=[indigo, very thick]
    \draw[edgeStyle] (A) -- (D);
    \draw[edgeStyle] (C) -- (D);
    \draw[edgeStyle] (A) -- (C);
    \tikzstyle{edgeStyle}=[indigo!45, very thick]
    \draw[edgeStyle] (D) -- (N);
    \draw[edgeStyle] (D) -- (R);
    \draw[edgeStyle] (C) -- (H);
    \draw[edgeStyle] (C) -- (F);
    \draw[edgeStyle] (C) -- (K);
    \draw[edgeStyle] (C) -- (R);
    \node[scale=1.4, below left] at (R) {2};
    \node[scale=1.4, above left] at (K) {1};
    \node[scale=1.4, above right] at (F) {1};
    \node[scale=1.4, below left] at (H) {1};
    \node[scale=1.4, above left] at (P) {1};
    \node[scale=1.4, below left] at (B) {1};
    \node[scale=1.4, above right] at (N) {1};

    \end{scope}
    \end{tikzpicture}
  \end{minipage}
  \vspace{0.3cm}
    \definecolor{niceblue}{rgb}{0.17, 0.3, 0.6}

  \begin{minipage}[b]{\Figonelen}


    \begin{tikzpicture}
    \begin{scope}[xscale=0.58, yscale=.5, transform shape]
    \drawGenerateLzero{niceblue}{niceblue!48}{niceblue!20};
    

    \node[scale=1.4, below left] at (B) {3};
    \node[scale=1.4, above right] at (C) {6};
    \node[scale=1.4, below left] at (D) {4};
    \node[scale=1.4, above left] at (P) {3};
    \end{scope}
    \end{tikzpicture}    
  \end{minipage}
  \hfill
  \begin{minipage}[b]{\Figonelen}

    \begin{tikzpicture}
    \begin{scope}[xscale=0.58, yscale=.5, transform shape]
      \drawGenerateLzero{niceblue}{niceblue!48}{niceblue!20};

    \node[client, fill=niceblue] at (C) {};
    \node[client, fill=niceblue!48] at (R) {};
    \node[client, fill=niceblue!48] at (H) {};
    \node[client, fill=niceblue!48] at (F) {};
    \node[client, fill=niceblue!48] at (K) {};
    
    \tikzstyle{edgeStyle}=[niceblue, very thick]
    \draw[edgeStyle] (A) -- (C);
    \tikzstyle{edgeStyle}=[niceblue!48, very thick]
    \draw[edgeStyle] (C) -- (D);
    \draw[edgeStyle] (C) -- (H);
    \draw[edgeStyle] (C) -- (F);
    \draw[edgeStyle] (C) -- (K);
    \draw[edgeStyle] (C) -- (R);
    
    \node[scale=1.4, below left] at (B) {3};
    \node[scale=1.4, below left] at (H) {5};
    \node[scale=1.4, above right] at (F) {2};
    \node[scale=1.4, below right] at (K) {2};
    \node[scale=1.4, below left] at (R) {3};
    \node[scale=1.4, below left] at (D) {4};
    \node[scale=1.4, above left] at (P) {3};

   \end{scope}
    \end{tikzpicture}
  \end{minipage}
  \hfill
  \begin{minipage}[b]{\Figonelen}


    \begin{tikzpicture}
    \begin{scope}[xscale=0.58, yscale=.5, transform shape]
    \drawGenerateLzero{niceblue}{niceblue!48}{niceblue!20};

    \node[client, fill=niceblue] at (C) {};
    \node[client, fill=niceblue] at (H) {};
    \node[client, fill=niceblue!48] at (R) {};
    \node[client, fill=niceblue!48] at (L) {};
    \node[client, fill=niceblue!48] at (J) {};
    \node[client, fill=niceblue!48] at (F) {};
    \node[client, fill=niceblue!48] at (K) {};
    
    \tikzstyle{edgeStyle}=[niceblue, very thick]
    \draw[edgeStyle] (A) -- (C);
    \draw[edgeStyle] (C) -- (H);
    \tikzstyle{edgeStyle}=[niceblue!48, very thick]
    \draw[edgeStyle] (C) -- (D);
    \draw[edgeStyle] (C) -- (F);
    \draw[edgeStyle] (C) -- (K);
    \draw[edgeStyle] (C) -- (R);
    \draw[edgeStyle] (H) -- (J);
    \draw[edgeStyle] (H) -- (L);
    \draw[edgeStyle] (H) -- (R);
    \draw[edgeStyle] (H) -- (K);
    
    \node[scale=1.4, below left] at (B) {3};
    \node[scale=1.4, above right] at (J) {1};
    \node[scale=1.4, below right] at (L) {1};
    \node[scale=1.4, above right] at (F) {2};
    \node[scale=1.4, below right] at (K) {2};
    \node[scale=1.4, below left] at (R) {3};
    \node[scale=1.4, below left] at (D) {4};
    \node[scale=1.4, above left] at (P) {3};

    \end{scope}
    \end{tikzpicture}
  \end{minipage}
  \caption{The structural decomposition phase, of generating (from left to right) the base layer $L_0$ in {\SL} (top, purple) and {\SLP} (bottom, blue) from an arbitrary starting node.
  At any given step, the next layer $L_1$ consists 
  of all neighbors of the $L_0$ nodes. The value next to each $L_1$-node indicates 
  its number of neighbors in $L_0$ (in \SL) or its total degree (in \SLP). 
  }
\label{fig:gen-L0}
\end{figure}

%% file: 011-figure-reach-SL.tex
\tikzset{dotted pattern/.style args={#1}{
		postaction=decorate,
		decoration={
			markings,
			mark=between positions 0.25 and 0.75 step 0.25 with {
				\fill[radius=#1] (0,0) circle;
			}
		}
	},
	dotted pattern/.default={1pt},
}

\newcommand{\drawLayers}[4]{

\node[scale=3.5, #4] at (-8, -7.2) {$L_0$};
\node[scale=3.5, #4] at (-4, -7.2) {$L_1$};
\node[scale=3.5, #4] at (-0.4, -7.2) {$L_{\geq 2}$};

\fill[#1!25]  (-8,0) circle  [x radius=1.5cm, y radius=3cm] node{};
\draw[thin, color=#1]  (-8,0) circle  [x radius=1.5cm, y radius=3cm] node{};

\fill[#2!25] (-4,0) circle  [x radius=1.5cm, y radius=5cm] node{};
\draw[thin, color=#2]  (-4,0) circle  [x radius=1.5cm, y radius=5cm] node{};


\fill[#3!25]  (0.5,5) circle  [x radius=2.3cm, y radius=1.8cm] node{};
\fill[#4!25]  (0,1) circle  [x radius=1.8cm, y radius=1.5cm] node{};
\fill[#4!25]  (-1.2,-1.2) circle  [x radius=0.6cm, y radius=0.7cm] node{};
\fill[#4!25]  (-.3,-4) circle  [x radius=1.5cm, y radius=2cm] node{}; 
\draw[thin, #3]  (0.5,5) circle  [x radius=2.3cm, y radius=1.8cm] node{};
\draw[thin, #4]  (0,1) circle  [x radius=1.8cm, y radius=1.5cm] node{};
\draw[thin, #4]  (-1.2,-1.2) circle  [x radius=0.6cm, y radius=0.7cm] node{};
\draw[thin, #4]  (-.3,-4) circle  [x radius=1.5cm, y radius=2cm] node{};

\tikzstyle{client}=[draw=none,circle,minimum size=3ex, inner sep=0, fill=#1,]
\node[client] (L0-1) at (-8,2) {};
\node[client] (L0-2) at (-8,0) {};
\node[client] (L0-3) at (-8,-2) {};

\tikzstyle{client}=[draw=none,circle,minimum size=3ex, inner sep=0, fill=#4,]    
\node[client] (L1-1) at (-4,3.5) {};
\node[client] (L1-2) at (-4,1.75) {};
\node[client] (L1-3) at (-4,0) {};
\node[client] (L1-4) at (-4,-1.75) {};
\node[client] (L1-5) at (-4,-3.5) {};


\tikzstyle{client}=[draw=none,circle,minimum size=3ex, inner sep=0, fill=#3,]

\node[client] (L21-1) at (-1,5.5) {};
\node[client] (L21-2) at (-1,4.5) {};
\node[client] (L21-3) at (0,6.2) {};
\node[client] (L21-4) at (0,4) {};
\node[client] (L21-5) at (1,6.5) {};
\node[client] (L21-6) at (1,4.5) {};
\node[client] (L21-7) at (1,3.5) {};
\node[client] (L21-8) at (2,5) {};

\tikzstyle{client}=[draw=none,circle,minimum size=3ex, inner sep=0, fill=#4,]

\node[client] (L22-1) at (-1,1) {};
\node[client] (L22-2) at (0,2) {};
\node[client] (L22-3) at (0,0.5) {};
\node[client] (L22-4) at (1,1.3) {};

\node[client] (L23-1) at (-1, -1.1) {};

\node[client] (L24-1) at (-1,-4) {};
\node[client] (L24-2) at (0,-3) {};
\node[client] (L24-3) at (-1,-5) {};
\node[client] (L24-4) at (0,-4.7) {};

\tikzstyle{edgeStyle}=[#1!40, dashed, thin]
\draw[edgeStyle] (L0-1) to [bend right] (L0-2);
\draw[edgeStyle] (L0-1) to [bend right] (L0-3);

\tikzstyle{edgeStyle}=[#1!70, thin]
\draw[edgeStyle] (L0-1) -- (L1-1);
\draw[edgeStyle] (L0-1) -- (L1-2);
\draw[edgeStyle] (L0-2) -- (L1-2);
\draw[edgeStyle] (L0-2) -- (L1-3);
\draw[edgeStyle] (L0-3) -- (L1-3);
\draw[edgeStyle] (L0-3) -- (L1-4);
\draw[edgeStyle] (L0-3) -- (L1-5);

\tikzstyle{edgeStyle}=[#4!70, dashed, thin]
\draw[edgeStyle] (L1-1) to [bend right] (L1-2);
\draw[edgeStyle] (L1-1) to [bend right] (L1-3);
\draw[edgeStyle] (L1-4) to [bend right] (L1-5);

\tikzstyle{edgeStyle}=[#4, thin]
\draw[edgeStyle] (L1-1) -- (L21-1);
\draw[edgeStyle] (L1-2) -- (L21-2);
\draw[edgeStyle] (L1-3) -- (L21-2);
\draw[edgeStyle] (L1-3) -- (L22-1);
\draw[edgeStyle] (L1-3) -- (L23-1);
\draw[edgeStyle] (L1-4) -- (L23-1);
\draw[edgeStyle] (L1-4) -- (L24-1);
\draw[edgeStyle] (L1-5) -- (L24-3);

\tikzstyle{edgeStyle}=[#4!70, dashed, thin]
\draw[edgeStyle] (L23-1) -- (L24-1);
\draw[edgeStyle] (L21-2) -- (L22-1);
\tikzstyle{edgeStyle}=[#3, thin]
\draw[edgeStyle] (L21-1) -- (L21-3);
\draw[edgeStyle] (L21-3) -- (L21-5);
\draw[edgeStyle] (L21-5) -- (L21-8);
\draw[edgeStyle] (L21-2) -- (L21-4);
\draw[edgeStyle] (L21-4) -- (L21-6);
\draw[edgeStyle] (L21-4) -- (L21-7);
\draw[edgeStyle] (L21-4) -- (L21-3);

\tikzstyle{edgeStyle}=[#4, thin]
\draw[edgeStyle] (L22-1) -- (L22-2);
\draw[edgeStyle] (L22-1) -- (L22-3);
\draw[edgeStyle] (L22-3) -- (L22-4);

\draw[edgeStyle] (L24-1) -- (L24-2);
\draw[edgeStyle] (L24-2) -- (L24-3);
\draw[edgeStyle] (L24-1) -- (L24-2);
\draw[edgeStyle] (L24-2) -- (L24-4);

\draw[thin, color=#1]  (-8,0) circle  [x radius=1.5cm, y radius=3cm] node{};
\draw[thin, color=#2]  (-4,0) circle  [x radius=1.5cm, y radius=5cm] node{};
\draw[thin, #3]  (0.5,5) circle  [x radius=2.3cm, y radius=1.8cm] node{};
\draw[thin, #4]  (0,1) circle  [x radius=1.8cm, y radius=1.5cm] node{};
\draw[thin, #4]  (-1.2,-1.2) circle  [x radius=0.6cm, y radius=0.7cm] node{};
\draw[thin, #4]  (-.3,-4) circle  [x radius=1.5cm, y radius=2cm] node{};
}

\begin{figure*}[]

\hfill
  \begin{minipage}[b]{0.22\textwidth}
    \begin{tikzpicture}
    \begin{scope}[scale=0.23, transform shape]
    \drawLayers{indigo}{lightgray}{lightgray}{lightgray};
    \draw[indigo, very thick] (L0-2) -- (L1-3) node{};
    \node[client, fill=indigo] at (L1-3) {};
    \node[scale=3.5, indigo] at (-8, -7.2) {$L_0$};
        \tikzstyle{client}=[draw=none,circle,minimum size=3ex, inner sep=0, fill=gray]
    \end{scope}
    \end{tikzpicture}    
    \vspace{.2cm}
  \end{minipage}
  \hfill
  \begin{minipage}[b]{0.22\textwidth}

    \begin{tikzpicture}
    \begin{scope}[scale=0.23, transform shape]
    \drawLayers{lightgray}{darkred}{lightgray}{lightgray};
    \draw[indigo, thin] (L0-2) -- (L1-3) node{};
    \draw[darkred, very thick] (L1-3) -- (L21-2) node{};
    \tikzstyle{edgeStyle}=[darkred, thin];
    \draw[edgeStyle] (L1-3) -- (L22-1);
    \draw[edgeStyle] (L1-3) -- (L23-1);
    \draw[darkred] (L1-3) -- (L21-2) node{};
    \tikzstyle{edgeStyle}=[darkred, dashed, thin]
    \draw[edgeStyle] (L1-1) to [bend right] (L1-3);
    \tikzstyle{client}=[draw=none,circle,minimum size=3ex, inner sep=0]
    \node[client, fill=indigo] at (L0-2) {};
   \node[client, fill=darkred] at (L1-3) {};
   \node[client, fill=darkred] at (L21-2) {};
    \node[scale=4, below, darkred] at (L1-3) {{$u$}};
    \node[scale=4, left=0.35, darkred] at (L21-2) {{$v$}};
    \node[scale=3.5, darkred] at (-4, -7.2) {$L_1$};
    \end{scope}
    \end{tikzpicture}
    \vspace{.2cm}
  \end{minipage}
  \hfill
  \begin{minipage}[b]{0.22\textwidth}

    \begin{tikzpicture}
    \begin{scope}[scale=0.23, transform shape]
    \drawLayers{lightgray}{lightgray}{teal}{lightgray};
    \tikzstyle{client}=[draw=none,circle,minimum size=3ex, inner sep=0]
    \node[client, fill=indigo] at (L0-2) {};
    \node[client, fill=darkred] at (L1-3) {};
    \node[client, fill=teal] at (L22-1) {};
    \draw[indigo, thin] (L0-2) -- (L1-3) node{};
    \draw[darkred, thin] (L1-3) -- (L21-2) node{};
    \node[scale=4, left=0.35, teal] at (L21-2) {{$v$}};
    \node[scale=4, below, teal] at (L21-7) {{$w$}};
    \node[scale=3.5, teal] at (-0.4, -7.2) {$L_{\geq 2}$};
    
        \tikzstyle{edgeStyle}=[teal]
\draw[edgeStyle] (L21-1) -- (L21-3);
\draw[edgeStyle] (L21-3) -- (L21-5);
\draw[edgeStyle] (L21-5) -- (L21-8);
\draw[edgeStyle] (L21-2) -- (L21-4);
\draw[edgeStyle] (L21-4) -- (L21-6);
\draw[edgeStyle] (L21-4) -- (L21-7);
\draw[edgeStyle] (L21-4) -- (L21-3);
\tikzstyle{edgeStyle}=[teal!70, dashed, thick]
    \draw[edgeStyle] (L21-2) -- (L22-1);
    \end{scope}
    \end{tikzpicture}
  \end{minipage}
  \hfill
  \begin{minipage}[b]{0.22\textwidth}

    \begin{tikzpicture}
    \begin{scope}[scale=0.23, transform shape]
    \drawLayers{lightgray}{lightgray}{teal}{lightgray};
    \tikzstyle{edgeStyle}=[lightgray!70, dashed, thin]
    \draw[edgeStyle] (L21-2) -- (L22-1);
    \node[scale=4, below, teal] at (L21-7) {{$w$}};
    \tikzstyle{edgeStyle}=[indigo!70, thin]
    \draw[edgeStyle] (L0-1) -- (L1-1) node{};
    \draw[edgeStyle] (L0-1) -- (L1-2) node{};
    \draw[edgeStyle] (L0-2) -- (L1-3) node{};
    \draw[edgeStyle] (L0-2) -- (L1-2) node{};
    \draw[edgeStyle] (L0-3) -- (L1-3) node{};
    \tikzstyle{edgeStyle}=[darkred!70, thin]
    \draw[edgeStyle] (L1-1) -- (L21-1);
    \draw[edgeStyle] (L1-2) -- (L21-2) node{};
    \draw[edgeStyle] (L1-3) -- (L21-2) node{};
    \draw[edgeStyle] (L1-3) -- (L22-1) node{};
    \draw[edgeStyle] (L1-3) -- (L23-1) node{};
    \node[scale=2.5, above right=1] at (L1-1) {\textbf{1}};
    \node[scale=2.5, above right=1] at (L1-2) {\textbf{2}};
    \node[scale=2.5, above right=0.6] at (L1-3) {\textbf{2/3}};
    \node[scale=3,right=3.6] at (L1-2) {\textbf{11/24}};
    
\tikzstyle{client}=[draw=none,circle,minimum size=3ex, inner sep=0, fill=indigo,]
\node[client] (L0-1) at (-8,2) {};
\node[client] (L0-2) at (-8,0) {};
\node[client] (L0-3) at (-8,-2) {};

\tikzstyle{client}=[draw=none,circle,minimum size=3ex, inner sep=0, fill=darkred,]    
\node[client] (L1-1) at (-4,3.5) {};
\node[client] (L1-2) at (-4,1.75) {};
\node[client] (L1-3) at (-4,0) {};
    \end{scope}
    \end{tikzpicture}
  \end{minipage}
  \hfill
  \caption{Sampling a node from $L_{\geq 2}$ in \textbf{\SL}. 
  We start by picking a uniform edge $L_0$ and $L_1$, let $u$ denote its $L_1$-endpoint. We next traverse a random edge from $u$ to $v \in L_2$, if one exists. We then fully explore the $L_{\geq 2}$-component $C$ containing $v$, choosing a uniformly random node $w \in C$. A final rejection step estimates how likely it is for the process to end at $w$.}
  \label{fig:reach-SL}
\end{figure*}
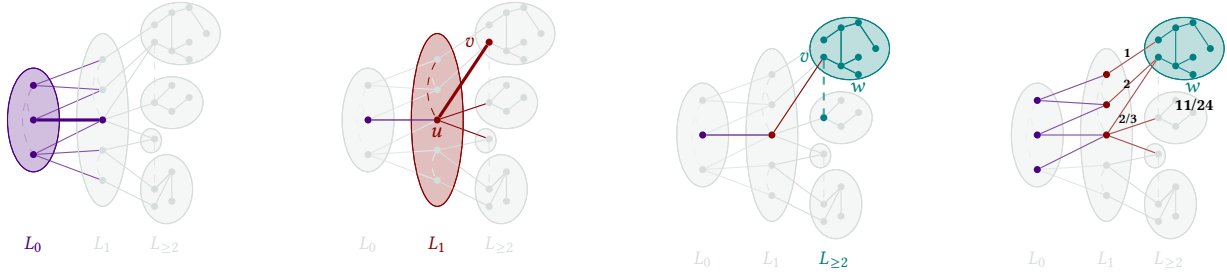

%% file: 013-theoretical-analysis.tex
\subsection{Convergence to 
Uniformity}
\label{sec:convergence}

Our main theorem states that samples generated by our algorithm converge to (near-)uniformity. The proof builds in part on the fact that our algorithm can estimate the size of $L_{\geq 2}$ given sufficient effort in the preprocessing phase. 
Proofs are given in  Section~\ref{sec:size_estimation}. 
Experiments validating the fast convergence of our size estimation procedures can be found in Section~\ref{size_est_par}.

\begin{theorem}
\label{thm:convergence_uniformity}
If our size estimation for $L_{\geq 2}$ 
is in $(1 \pm o(1))|L_{\geq 2}|$, and if the baseline reachability $rs_0$ used in our algorithm is the $o(1)$-percentile in the reachability distribution, then the output node distribution of\ \Sample\ is $o(1)$-close to uniform in total variation distance.
\end{theorem}
We stress that even in the case that the $L_0$ generation process is unsuccessful (in a sense that it does not break the $L_{\geq 2}$ vertices into small components),  it \emph{always} holds that our algorithm returns a close to uniform vertex, provided that the size and reachability estimates are correct. That is, the correctness of our algorithm holds for \emph{any} given $L_0$ (with high probability), and only the query complexity of subsequent sampling might be negatively affected, e.g., due to a high expected component size value.

%% file: 014-query_complexity.tex
\subsection{Query Complexity}
\label{sec:theoretical} \label{sec:query_comp}

\label{sec:analysis}
In this section, we analyze the query complexity of the sampling phase of our approach.
We show here that the query complexity of sampling nodes using $\SL$ is bounded as a function of several parameters related to the layered structure we maintain. Later on, we empirically show that the relevant parameters are indeed well-behaved in the networks we investigated. 
The starting point of our analysis is immediately after the preprocessing phase is completed. In particular, $L_0$ and $L_1$ are already known, as well as a good estimate of the size of $L_{\geq 2}$.
In addition, we have the ability to sample uniformly random edges between $L_0$ and $L_1$ without making any queries.
The assumptions we make 
are as follows.
\begin{itemize}
	\item\textbf{Reachability distribution.} We assume that the reachabilities of nodes in $L_2$ are relatively balanced: the reachability score $rs(v)$ of every $v \in L_2$ satisfies $rs_0 \leq rs(v) \leq c \cdot rs_0$, where $rs_0$ is viewed as the ``base reachability'', and $c > 1$ is not large. We empirically verify this in Section \ref{sec:other_exps}.
	\item\textbf{Entry points.} 
	Let $\alpha$ denote the fraction of edges $e$ between $L_0$ and $L_1$, for which the $L_1$-endpoint of $e$ has neighbors in $L_2$. Then $\alpha$ is precisely the probability that a single attempt at reaching $L_{\geq 2}$ succeeds (without taking the rejection step into account). In practice, $\alpha$ is known to be well-behaved \cite{RombachCore2014}, as most bridges to $L_{\geq 2}$ occur at higher-degree nodes of $L_1$.
	\item\textbf{Component sizes.} Set $w = \mathbb{E}[\,|CC(v)|\,] $, where $v \in L_{\geq 2}$ is (distributed as) the result of a single run of our procedure \SLtwo, and $CC(v)$ is the $L_{\geq 2}$-component in which $v$ resides.
	Intuitively, $w$ measures the sizes of 
	components
	that we reach, and we empirically validate that it is typically small on both synthetic and real-world networks, see Section \ref{sec:other_exps}.
    \item\textbf{Degrees of component nodes.} We assume that for all components $C$ of $L_{\geq 2}$, the number of bridges from $C$ to the rest of the network is at most $d \cdot |C|$, for a small integer $d$. This is in line with the well-observed fact \cite{leskovec, RombachCore2014} that peripheral components are weakly connected to the network.
\end{itemize}
Our experiments verify that the parameters discussed here are indeed well-behaved when the size $\ell_0$ of $L_0$ is chosen correctly -- see Section \ref{sec:other_exps} for more details.
We bound the expected query complexity of our sampling algorithm as a function of the above parameters. Crucially, this implies that, once the preprocessing phase is complete, the query complexity does not directly depend on the network size or on the mixing time of long random walks.
\begin{theorem}\label{thm:qcompl}
The expected query complexity of sampling a single node using \SL is
$
O\left(
c \cdot \left(\frac{1}{\alpha} + w d\right)\right)
$.
\end{theorem}
\begin{proof}
Sampling from either $L_0$ or $L_1$ requires no queries. Hence, consider sampling  from $L_{\geq 2}$.
Each invocation of \SLtwo\ returns a node in $L_{\geq 2}$. However, due to the rejection sampling in Step 3(b) of \Sample, some of the samples are discarded, and the while loop is repeated. The probability a sample is accepted is at least $1/c$. 

Now consider a single invocation of \SLtwo.
The expected query complexity of the procedure stems from the number of attempts it takes to reach from $L_1$ to $L_{2}$, and once it reaches some $v \in L_2$, from computing $CC(v)$ and its reachability score.
By definition, the probability of reaching from $L_1$ to $L_2$ is $\alpha$, so the expected number of attempts is $1/\alpha$. Once in $L_2$, in order to compute $CC(v)$ we need to traverse the component,  and for each node in the component, query its neighbors (to determine if they belong to the component or not). Therefore, the expected query complexity of this step is $O(w\cdot d)$. Hence, the expected complexity of a single invocation is $O(1/\alpha+w\cdot d)$.

All in all, we get that the expected sample complexity is bounded by the expected complexity of a single invocation $O(1/\alpha+w\cdot d)$, divided by the minimum possible success probability of a single invocation, $1/c$. The resulting bound is $O( c (1/\alpha + w\cdot d))$.
\end{proof}

%% file: 015-experiments.tex
\begin{table}[b]
\centering
\begin{tabular}{@{}lcccccccc@{}}
\toprule
Dataset & $n$ & $m$ & $d_{\text{avg}}$ & 
\multicolumn{2}{c}{$L_0$ size} \\ \toprule 
& & & & SL & SL+ \\ \toprule
Epinions \cite{epinions2003} & 76K & 509K  & 13.4 & 3K & 1K \\ \midrule
Slashdot \cite{leskovec} & 82K & 948K & 23.1 & 3K & 2K \\ \midrule
DBLP \cite{Yang2015} & 317K & 1.05M & 6.62 & 30K & 20K \\ \midrule
Twitter-Higgs \cite{Higgs13} & 457K & 14.9M & 65.1 & 25K & 10K \\ \midrule
Forest Fire~\cite{FF1, FF2} & 1M & 6.75M & 13.5 & 10K & 10K \\ \midrule
Youtube \cite{Yang2015} & 1.1M & 2.99M & 5.27 & 30K & 10K \\ \midrule
Pokec \cite{takac2012data} & 1.6M & 30.6M & 37.5 & 200K & 100K \\
\midrule
SinaWeibo \cite{zhang2014characterizing} & 58.7M & 261M & 8.91 & 500K & 100K \\
\toprule\\ 
\end{tabular}
\begin{minipage}{0.96\columnwidth}
\caption{
The list of networks we considered with numbers of nodes ($n$), edges ($m$), their average degrees ($d_{\text{avg}}$), and $L_0$ sizes we selected for \SL and \SLP.
}\label{table:exp}
\end{minipage}
\end{table}
\input{610-RW-comparisons/RW-comps}

In this section, we describe several experiments we conducted, comparing our algorithms to previous approaches which are all based on random walks (Section \ref{subsec:comparison_exp}), and explaining the query-efficiency of our methods (Section \ref{sec:other_exps}). 

\subsection{Evaluation of Query Complexity}
\label{subsec:comparison_exp}
The main experiment computes the amortized number of queries per sample of our algorithm, and compares it with the corresponding query complexity of existing RW-based approaches. 
In the standard query model, we compare our algorithm \SL\ with two random walk-based algorithms, Rejection sampling (\rej) and Metropolis-Hastings (\mh). 
In the stronger query model, we compare \SLP{} to Metropolis-Hastings ``plus'' (\mhp). The methods \rej, \mh, and \mhp were all described in detail by Chierichetti et al.~\cite{ChiericettiWWW2016}. 
 In  \rej, the algorithm  performs a standard (unbiased) random walk, where nodes are subject to rejection sampling according to their degree; in \mh the neighbor transition probabilities are controlled by the neighbors' degrees. \mhp is the same as \mh, but assumes the stronger query model, where  a node query also reveals the degrees of its neighbors.
RW-based algorithms are most commonly used to sample multiple nodes by performing a long random walk, and sampling a new node once every fixed interval to allow for re-mixing. Indeed, as discussed in Section \ref{sec:lower_bound}, to ensure that the node samples will be uniform and independent, the interval length must allow the walk to mix between subsequent samples.

\para{Setting for our algorithm}
We examine seven online social and information networks of varying sizes and characteristics, taken from widely used network repositories \cite{snapnets, nr}.
We also examine our algorithm on a network generated by the Forest Fire model with parameters $p_f=0.37$ and $p_b=0.3$, which are standard for this model~\cite{FF1}. 
The networks, along with their basic properties, are described in Table~\ref{table:exp}.
  For each network, we performed a small grid search to obtain a reasonable value for the input parameter $\ell_0$ (the target size of $L_0$) in our algorithm. The values we used for each network are given in Table~\ref{table:exp}. For the other two input parameters, $s_1$ and $s_{\geq 2}$,  we observed that choices of $3{,}000$ and 200 respectively are generally sufficient for \SL on the first seven networks (for Epinions and Slashdot, we picked $s_1 = 1{,}000$). In \SLP, values of $s_1 = 1{,}000$ and $s_{\geq 2} = 100$ are generally sufficient for the seven smaller networks.
Separately, for SinaWeibo we picked larger values, of $s_1=30k$ and $s_{\geq 2}=3k$ for both \SL{} and \SLP, since the network is substantially larger.

We ran each of our algorithms \SL and \SLP for 5-10 times on each of the eight networks; the amortized query complexity we calculated is the average over these runs.
As part of our pipeline, we verified the quality of our solution by configuring the algorithm's parameters so as to ensure that the samples generated by our algorithm are close to uniform. Specifically, we fixed a small empirical threshold $t$ ($0.01$ in most cases) and parameters $s_1$ and $s_{\geq 2}$ as above, while varying the value of the error parameter $\eps$ in our algorithm. For each choice of $\eps$, we ran the following for 10 times: we sampled $n$ nodes using our algorithm (parameterized by $\eps$), where $n$ is the graph size. In each of the runs, we calculated the empirical distance to uniformity; if the average empirical distance over the 10 runs is more than $t$ away from the expected value for a true uniform distribution, $\eps$ is discarded. Thus, our final choice of $\eps$ ensures near-uniformity of the output samples.

\para{Setting for random walks}
As mentioned above, the most standard approach to sampling multiple nodes using a random walk is by running a single long walk and extracting samples from this walk in fixed intervals. We examined this approach in two phases: setting the interval length, and evaluating the query complexity in view of this choice of interval length.

We judiciously set interval lengths that allow for proper mixing. This is explained in detail in Appendix \ref{sec:interval_length}, but briefly, 
we generated a large number of short random walks from the same starting point and evaluated at what point in time these walks mix. To this end,  we computed in each time step, for all walks simultaneously, the empirical distance to uniformity (using the same value of the empirical threshold $t$ as in our algorithm) or the number of collisions, which is also an indicator of distance to uniformity~\cite{goldreich2011testing}. To ensure the variance is controlled, we ran this procedure from 3-5 different starting points for each of the algorithms \rej, \mh, \mhp.

To compute the amortized query complexity, we ran the random walk algorithms on each of the networks, while keeping track of the cumulative number of queries. Then, we computed the mean number of queries per sample as the walk progressed. 

\para{Main results} Figure~\ref{fig:RW-comps} depicts the comparison results
for the six smaller real-world networks. Figure~\ref{fig:sinaweibo} and~\ref{fig:FF-comparison} show the results for SinaWeibo, and for a Forest Fire generated network, respectively. 
The number of node samples  in each of the first six networks, as well as in the FF one, is between $0.1\%$ and $1\%$ of the total number of nodes. While the results at the lower end, $0.1\%$, show the relatively steep initial price of the structural learning phase of our algorithm, the higher end of our sample size clarifies the stark differences in performance between the methods. In SinaWeibo, due to its sheer size, we considered a wider interval, from $10k$ samples (less than $0.02\%$ of the nodes) to about $600k$ samples ($1\%$).

As is evident in the plots, \SL and \SLP obtained significantly improved results compared to their RW-based counterparts. In all cases, and throughout the runs (as more samples are gathered), \SL demonstrated a query complexity that in all cases offers query complexity savings of at least $40\%$, and often much more, compared to both \rej and \mh. Moreover, \rej consistently required fewer queries on average than its counterpart \mh. This is in 
line with previous results from \cite{ChiericettiWWW2016}. In the more powerful query model, \SLP also gave at least $50\%$ (and almost always better) improvement over its random walk analog $\mhp$. The most dramatic improvement was for  SinaWeibo, the largest network, where \SL and \SLP yielded reductions reaching $90\%$ and $95\%$ in the query complexity, respectively, compared to their random walk counterparts. 
Curiously, as shown in Figure \ref{fig:sinaweibo}, the query complexity of \SLP in SinaWeibo was in some cases less than one query per sample. While this may seem counter-intuitive at first, we note that node samples from $L_1$ are  generated by  our algorithm without any query cost.
One feature of SinaWeibo that we observed is that a large majority of the nodes in the network are located in $L_1$, even for small $L_0$ sizes. Thus, many samples do not induce any query-cost.

Interestingly, our algorithm was challenged by the DBLP network, which required a costly $L_0$-construction stage, resulting in an initial disadvantage. However, as is clear in the figure, our algorithm recovers at samples of at least $1{,}000$ nodes, to quickly reach consistent improvement of about $40\%$ compared to \rej. 
We believe that these difficulties stem from the fact that in DBLP, a collaboration network, nodes have a weaker tendency to connect to very high degree nodes than in most social or information networks.

\subsection{Other Experiments}
\label{sec:other_exps}
\para{Structural layering parameters}
In Section \ref{sec:theoretical} we have seen that two factors mostly control the query complexity of our sampling phase: the typical size of $L_{\geq 2}$-components,      
and the reachability distribution of nodes in $L_{\geq 2}$. We empirically demonstrate that for the ``right'' size of $L_0$, these two factors are well-behaved, explaining the improved query complexity of our method.

Define $\mu$ as the average component size of a node in $L_{\geq 2}$.
This is a weighted (biased) average, giving more weight to the larger components. To see why this weighted average is of interest, recall our definition of $w$ from Section \ref{sec:analysis}: the expected size of a component reached by the procedure \SLtwo.
Under the assumption that all node reachabilities in $L_{\geq 2}$ are of the same order, it holds that $w = \Theta(\mu)$. 
To analyze the typical size of reached components in $L_{\geq 2}$, we computed the value of $\mu$ as a function of the $L_0$-size. This was done five times for each network, and is demonstrated for DBLP, Youtube, Twitter-Higgs, and SinaWeibo in Figure \ref{fig:L0-L2-exps} (right). Interestingly, the weighted average $\mu$ seems to decay exponentially as a function of the $L_0$-size (note that the $y$-axis here is log-scaled), until it converges to a constant. To the best of our knowledge, this exponential decay was not previously observed in the literature, and we believe it warrants further research; as the results show, the rate of decay differs majorly between different networks, and explains the choices of $L_0$-size that we made for the different networks: ideally, one would like $L_0$  to be small, while inducing a small enough $\mu$-value (say, $10$ or $20$). \label{par:L0-size}

In Figure \ref{fig:L0-L2-exps} (left), we investigated what minimal size of $L_0$ is required in a Forest Fire graph in order to satisfy $\mu \leq 10$. The experiment was run for different values of the graph size $n$, while fixing the FF parameters $p_f = 0.37$ and $p_b = 0.3$ as before. The results show that the required $|L_0|$ is clearly sublinear in $n$: they decay from about $3.3\%$ for $n=1$K to about $2\%$ for $n=100$K.

\para{Reachability distribution}
In our theoretical analysis, we claim that if the reachabilities of nodes in $L_{\geq 2}$ are all roughly of the same order, then the rejection step of our algorithm is not too costly. Here, we demonstrate that this assumption on the reachability distribution indeed approximately holds in practice. Specifically, Figure~\ref{fig:reach_dist} presents the reachability distribution of the $L_{\geq 2}$-nodes in DBLP for $|L_0|$ of $30k$ (the distributions for other networks are similar). For clarity, we discarded the top 3\% reachabilities, which form a thin upper tail, and only show the lower 97\% here. As can be seen, indeed most reachabilities are roughly of the same order: almost all $L_{\geq 2}$-nodes in the experiment have reachability up to 1, where a majority of them are between roughly 0.05 and 0.2. 
\begin{figure}[]
\centering
 \hspace{-25pt}
    \def\imagelen{.7\columnwidth}
  \includegraphics[width=\imagelen]{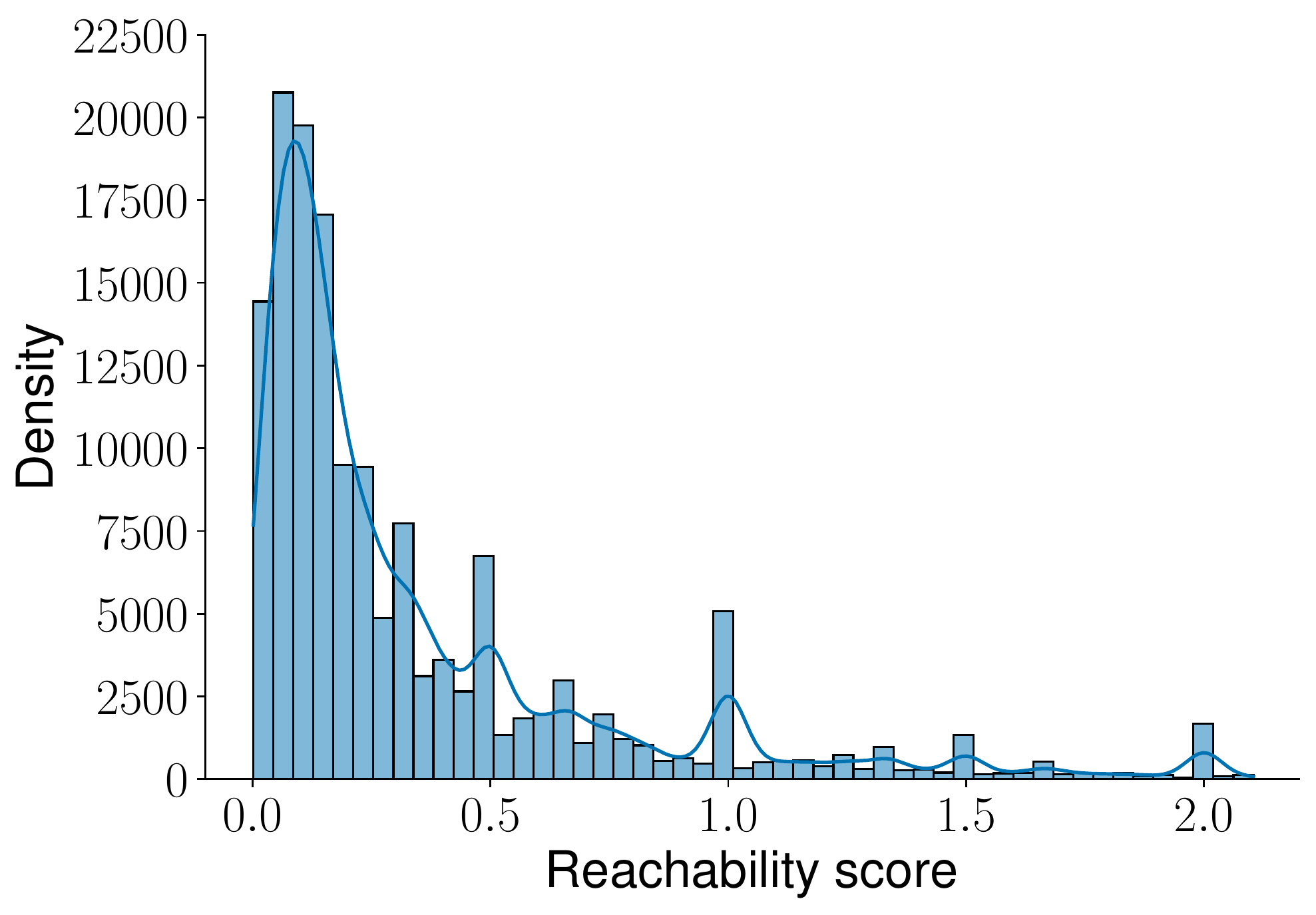}
  \caption{Distribution of node reachabilities in $L_{\geq 2}$ for DBLP.
  }
  \label{fig:reach_dist}
\end{figure}

\para{Size estimation}\label{size_est_par}
As mentioned, our algorithm needs to compute an accurate size estimate for $L_{\geq 2}$ as part of its preprocessing.
We next empirically demonstrate the quick convergence of our size estimate as a function of the numbers of nodes from $L_1$ and $L_{\geq 2}$ we visit during the preprocessing.
Recall the size estimation procedure \CompLTsize{} described in Section \ref{sec:alg_desc} (see also Section \ref{sec:preprocess_missing}).
Here we demonstrate the quick rate of convergence of this procedure, validating our choices of the parameters $s_1$ and $s_{\geq 2}$. 

The size of $L_{\geq 2}$ satisfies the following: $|L_{\geq 2}| = |L_1| \cdot d_1^+/d_2^-$, where $d_1^+$ is the average, over all nodes $v \in L_1$, of the number of neighbors of $v$ in $L_2$; and $d_2^-$ is the symmetric quantity, i.e.~the (unweighted) average over all nodes in $L_{\geq 2}$ of their number of neighbors in $L_1$.
\CompLTsize\ estimates $d_1^+$ and $d_2^-$ by taking $s_1$ samples from $L_1$ and $s_{\geq 2}$ sampled nodes from $L_{\geq 2}$ (using \SLtwo, without rejection; as these are biased, we use weighted averaging). It then uses these estimates to compute an estimate of $|L_{\geq 2}|$.
In the current experiment, we separately check how the chosen values of $s_1$ and $s_{\geq 2}$ affect the size estimate of $L_{\geq 2}$. 
We run two experiments, each of them five times for each of the networks; see the results for DBLP in Figure \ref{fig:size-estim-SL-supp} as a representative example. 
In the first experiment, we fix the value of $d_2^-$, and measure the error in the estimate  $\bar \ell_{\geq 2}$ when $d_1^+$ is computed as the average out-degree over $s_1$ samples.
The second experiment is similar, except that $d_1^+$ and $d_2^-$ switch roles: $d_1^+$ is fixed to its actual value, whereas we compute an estimate of $d_2^-$ from $s_{\geq 2}$ nodes in $L_{\geq 2}$ obtained through our algorithm 
(without the rejection step), where each reached node is assigned a weight that is inversely proportional to its reachability. As Figure~\ref{fig:size-estim-SL-supp} shows, it suffices to take $s_1$ of order a few thousands (left) and $s_{\geq 2}$ of order a few hundreds (right) to obtain a small error in the size estimation.

\begin{figure}
\centering
    \def\imagelen{0.48\columnwidth}
  \hspace{-.7cm}
  \includegraphics[width=\imagelen]{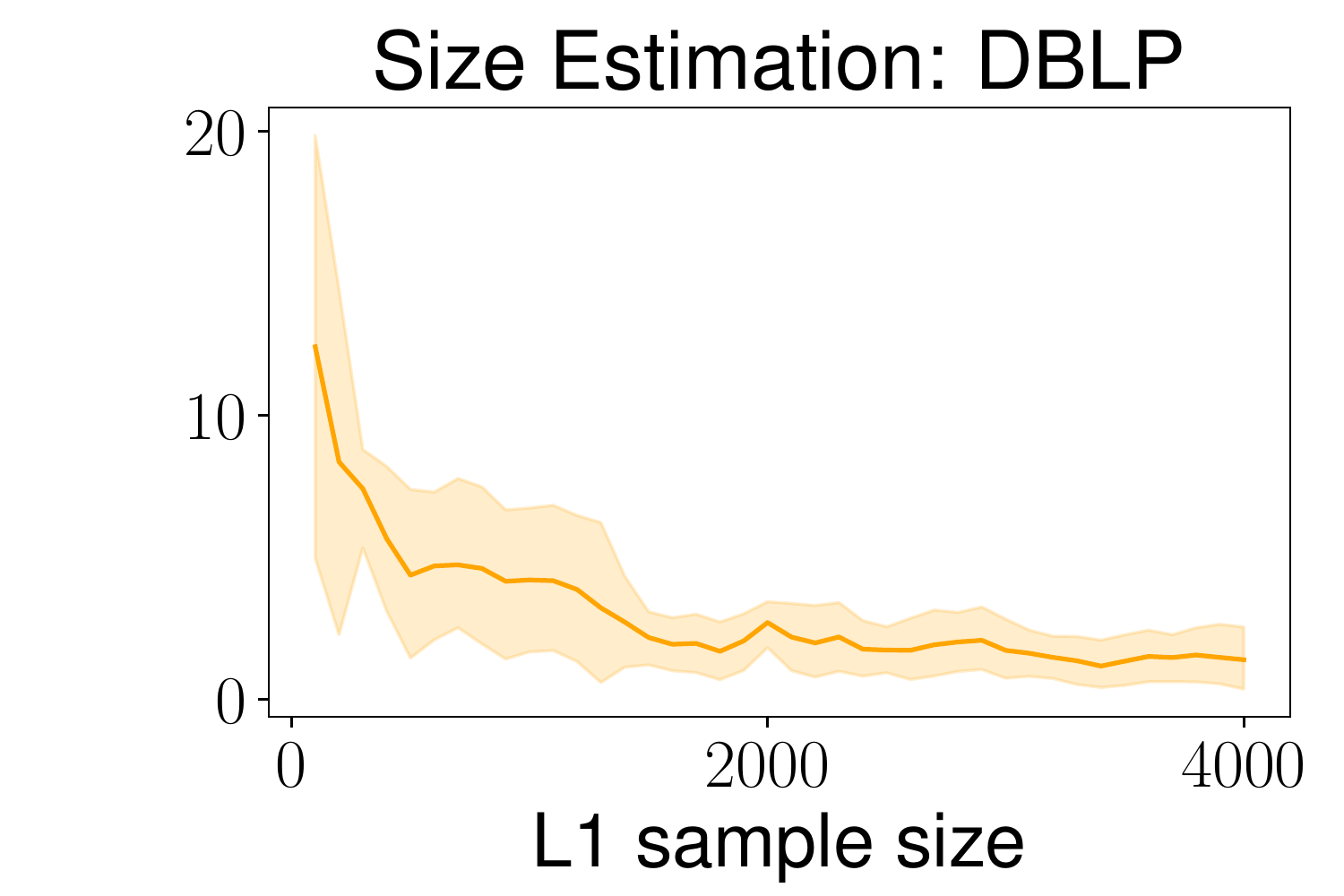}
  \includegraphics[width=\imagelen]{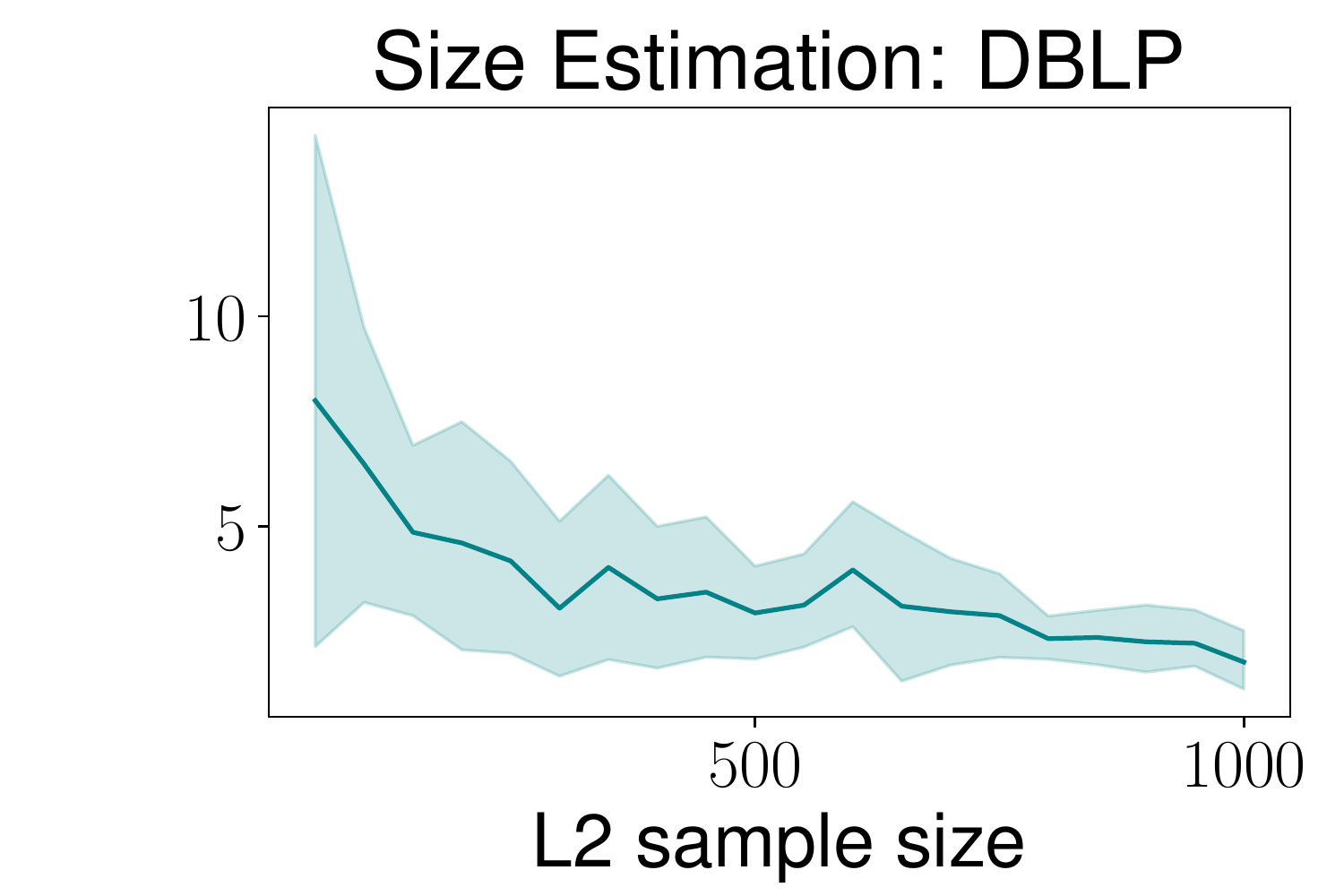}
  \vspace{.2cm}
  \caption{Error (\%) of the size estimate obtained by \SL for DBLP during our structural decomposition phase, as a function of the number $s_1$ of nodes queried from $L_1$ (left) and the number $s_{\geq 2}$ of reached nodes from $L_2$ (right).
  }
  \label{fig:size-estim-SL-supp}
\end{figure}

%% file: 610-RW-comparisons/RW-comps.tex
\begin{figure*}[!ht]
\centering
  \def\imagelen{0.33\linewidth}
  \includegraphics[width=\imagelen]{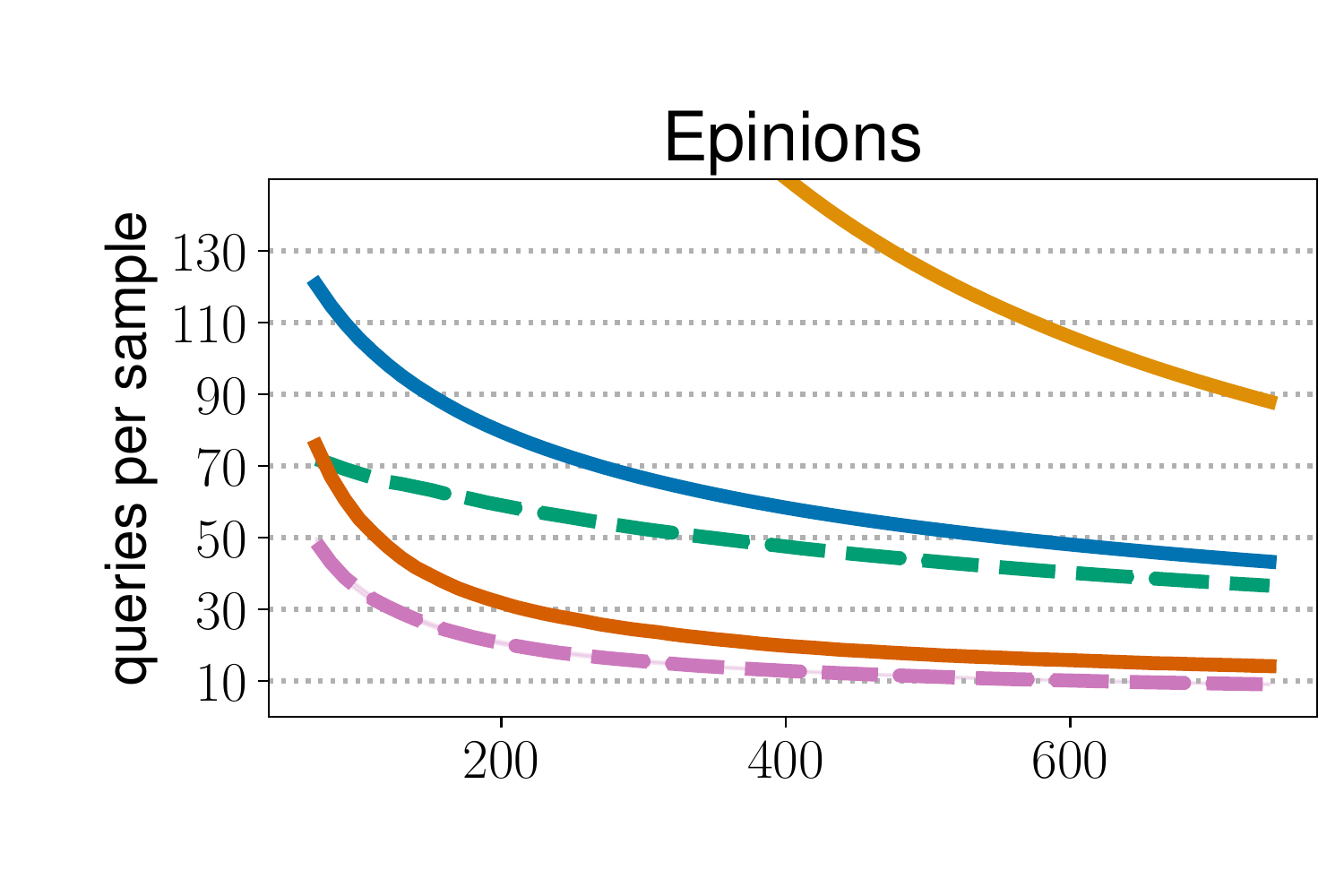}
    \vspace{-12pt}
  \includegraphics[width=\imagelen]{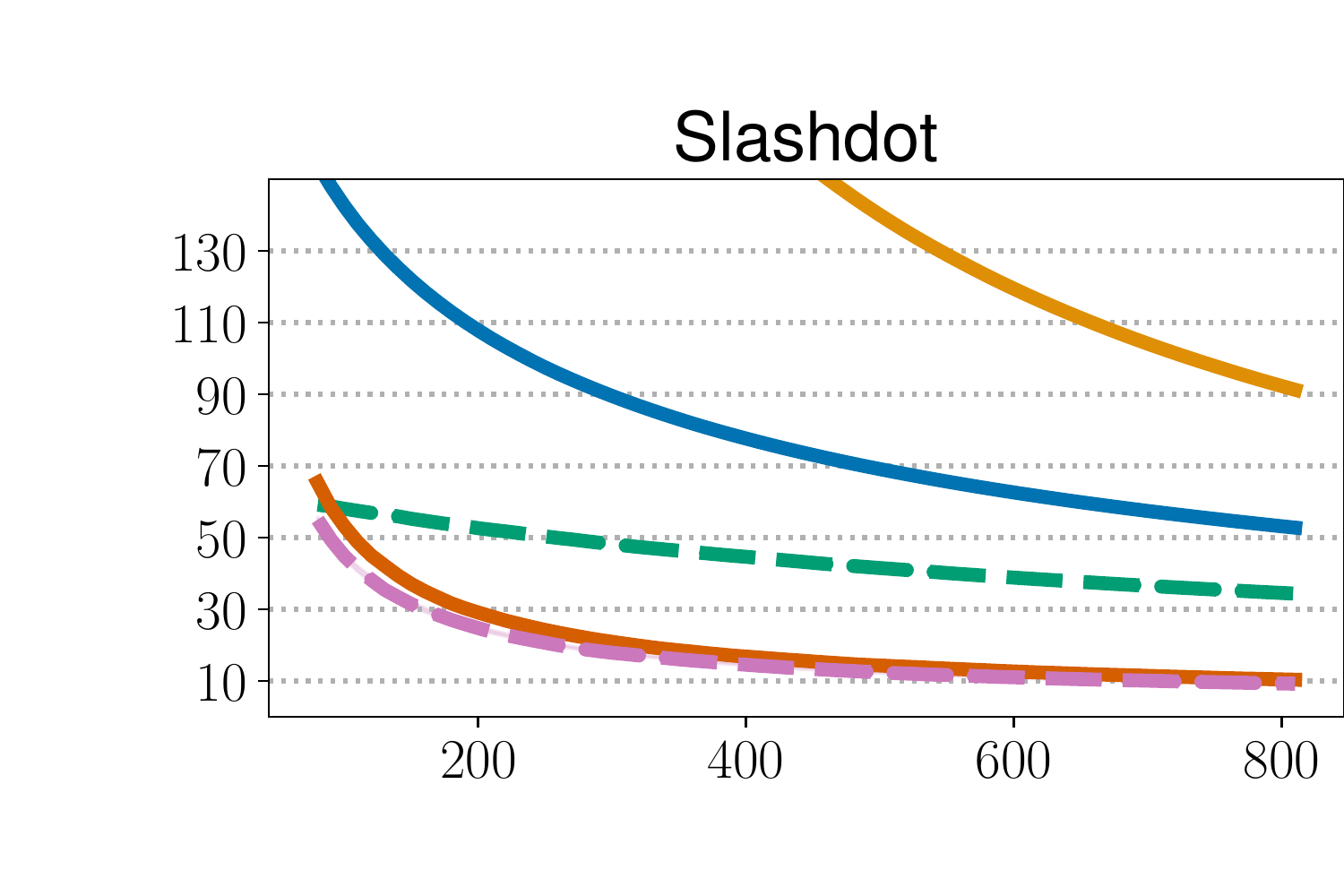}
  \vspace{-12pt}
  \includegraphics[width=\imagelen]{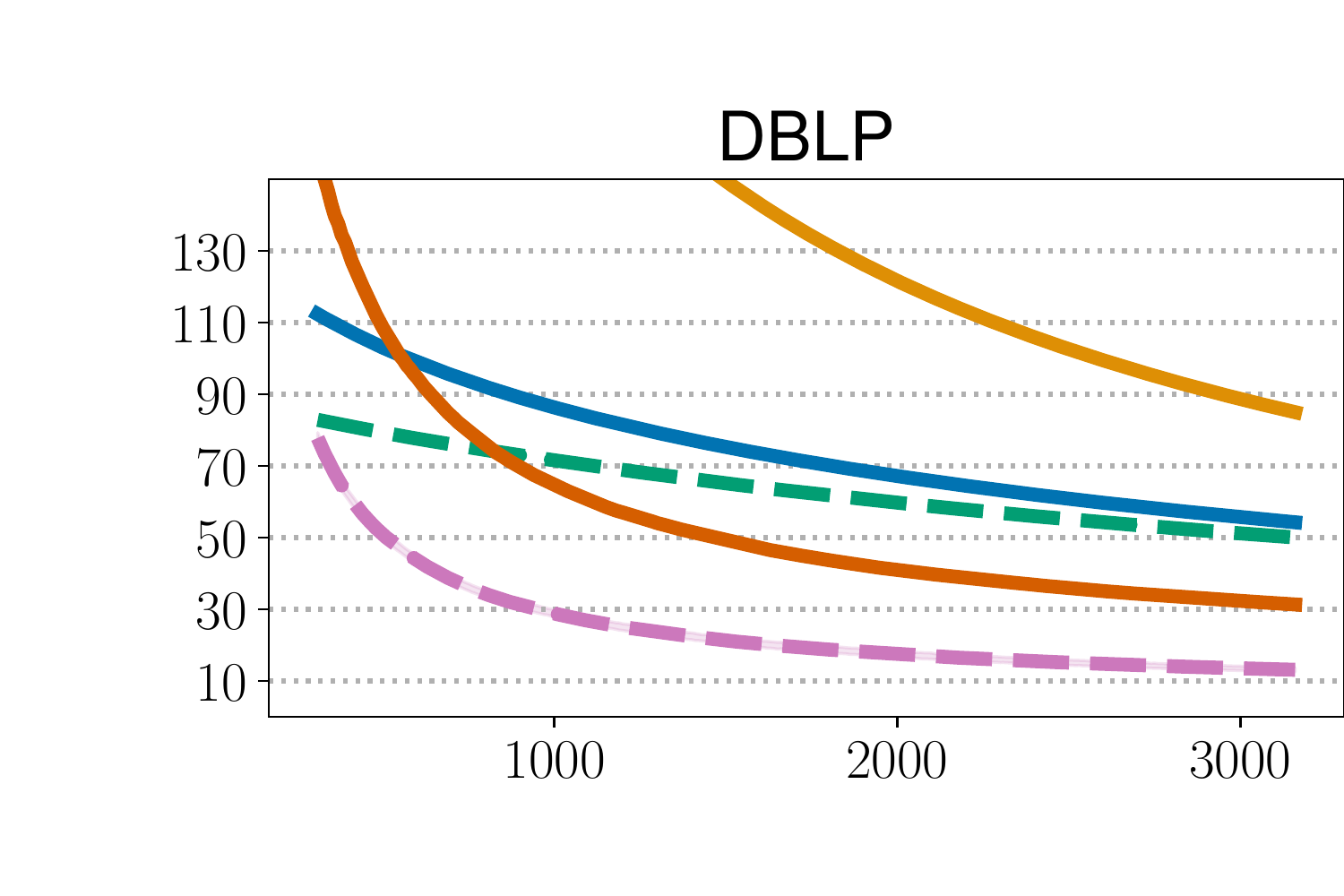}
  \vspace{-12pt}
  \end{figure*}
  \begin{figure*}[!ht]
  \centering
  \def\imagelen{0.33\linewidth}
  \includegraphics[width=\imagelen]{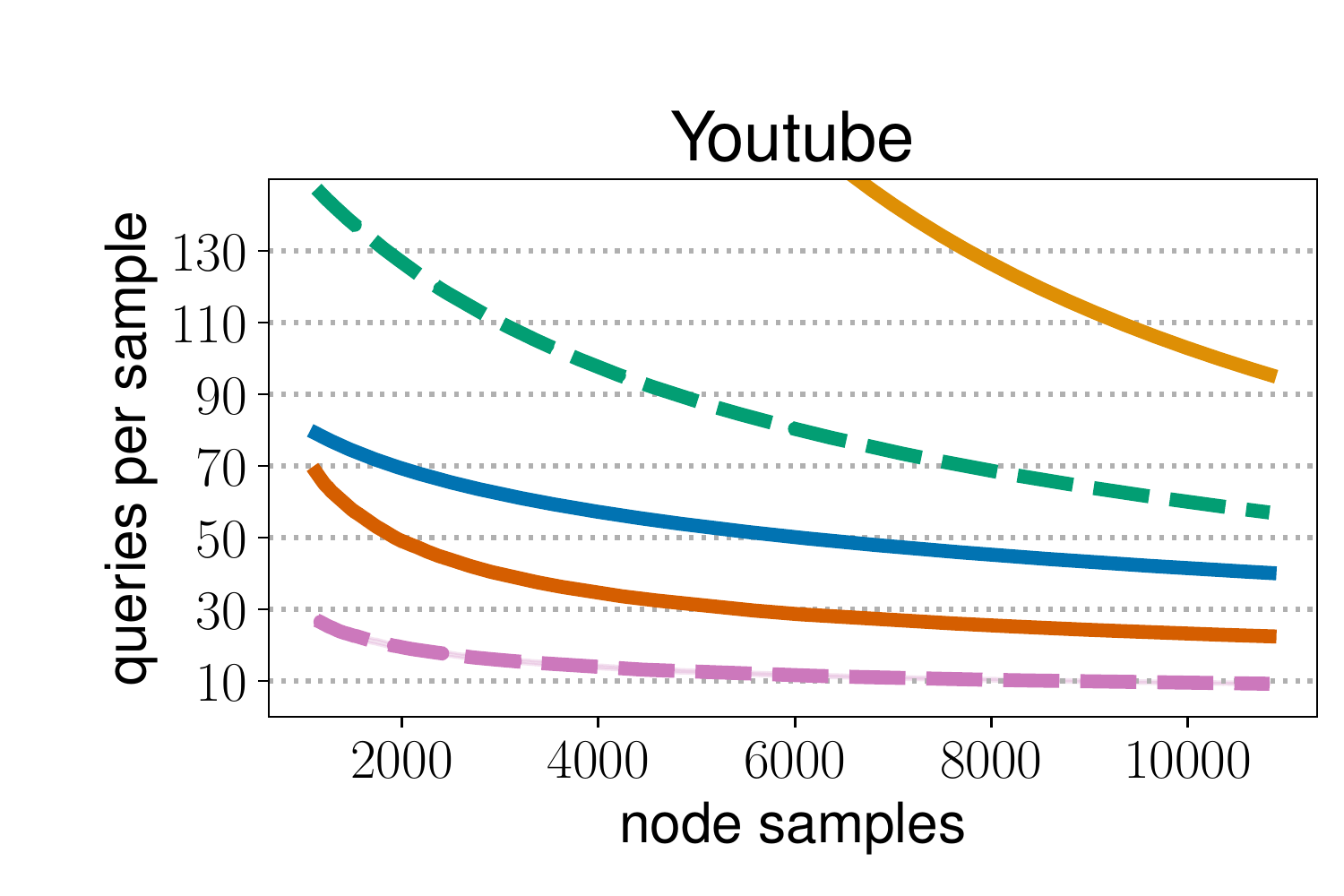}
  \includegraphics[width=\imagelen]{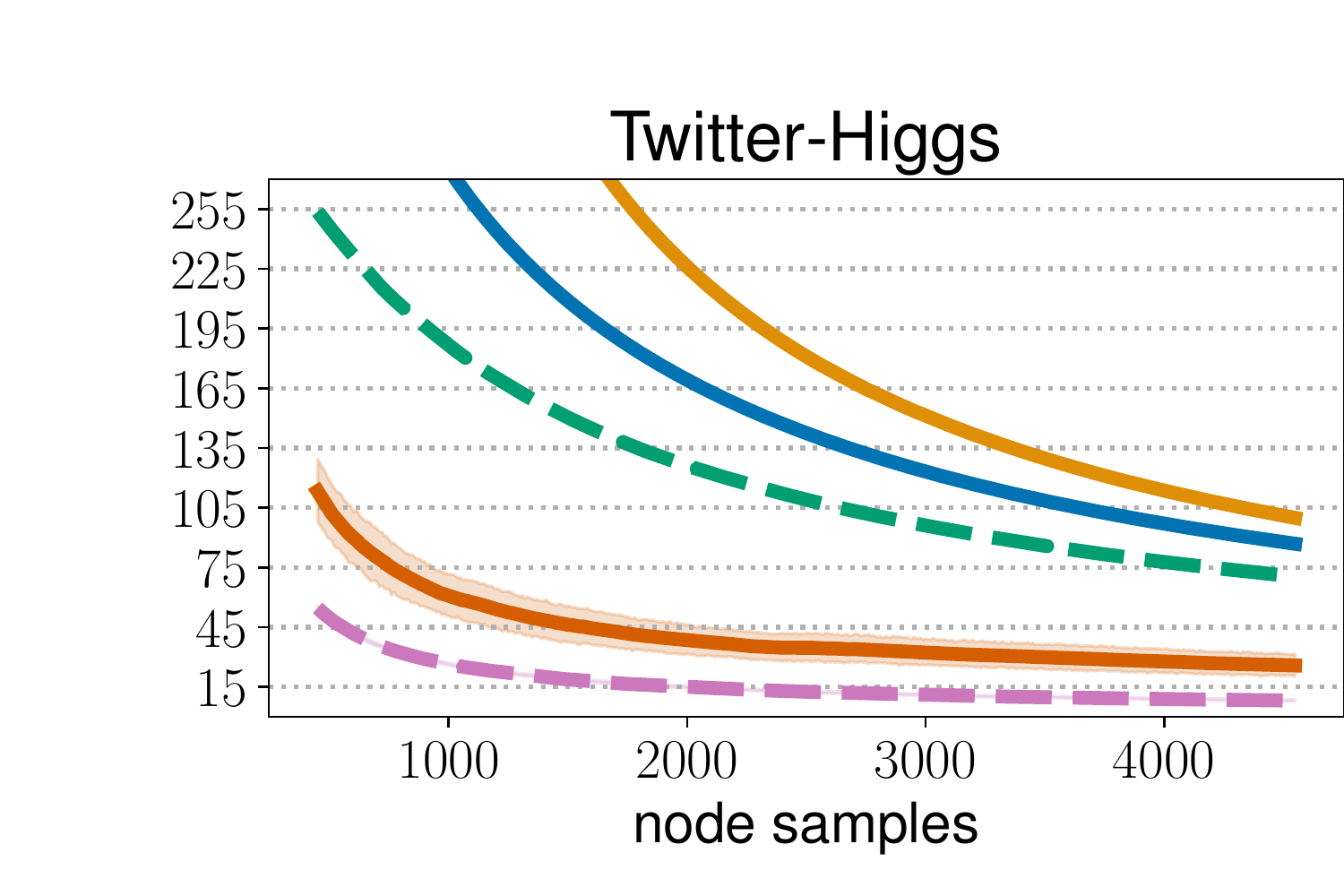}
  \includegraphics[width=\imagelen]{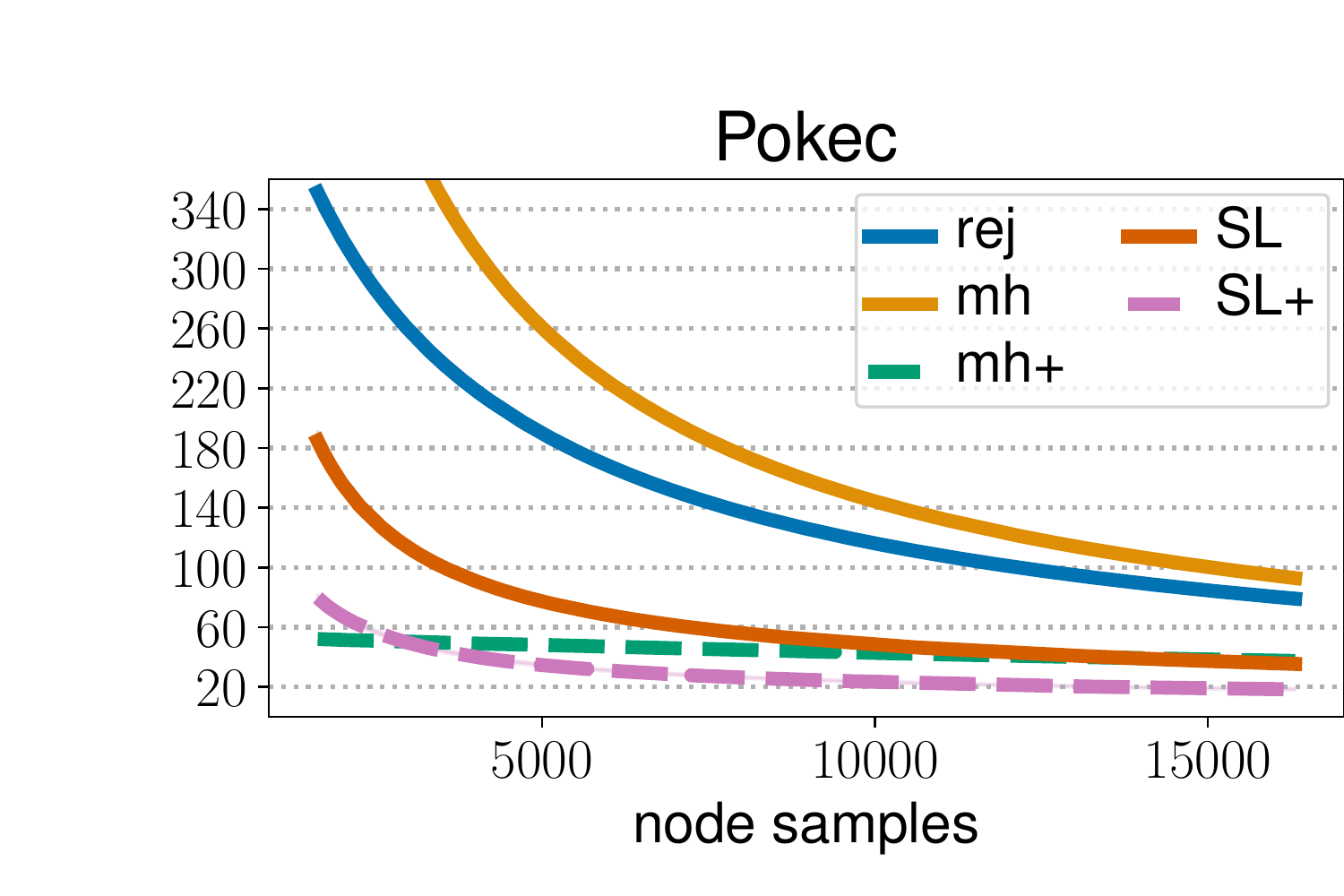}
  \caption{Amortized query complexity of our methods  compared to random-walk-based node sampling methods.
  In the standard query model, \SL{} is compared against Rejection Sampling (\rej) and  Metropolis Hastings (\mh), and in the stronger query model, \SLP{} is compared against Metropolis Hastings ``plus'' (\mhp).
   For each of the networks, the number of node samples ranges  between $0.1\%$ and $1\%$ of the network's order. 
  See Section \ref{sec:experiments} and Table \ref{table:exp} for details. 
  }
\label{fig:RW-comps}
\vspace{5pt}
\end{figure*}

%% file: 016-summary.tex
\section{Conclusions and open questions}

We have presented an algorithm that supports query-efficient sampling  of multiple nodes in a large social network, exhibiting the 
efficiency of our algorithm  compared to the state of the art on a variety of datasets, that include up to tens of millions of nodes.
We gave theoretical bounds on our algorithm's complexity in terms of several graph parameters. We then empirically confirmed that, in all the social networks we have examined, these graph parameters are well-behaved.
One major concern is the question of generality. That is, what is the scope of networks  for which our methods are suitable. 
As stated in the experimental section, our algorithms gave better or comparable results on all social networks we have tested, and the good performance on the Forest Fire model, a realistic generative model, further hints to wide applicability.

Our algorithm has some disadvantages compared to random walks. First, it is substantially more complicated, and its analysis does not directly depend on one structural parameter of the network (such as the mixing time for random walks). Second, it inherently relies on a data structure which has a high memory consumption, and may not always be suitable in situations that require bounded-memory algorithms. Third, it depends on the size of $L_0$, which differs between different networks, so it requires fine tuning according to the given network. Finally, it exploits unique  properties of social networks. These properties do not hold for bounded degree graphs, or for, say, real-world road networks. For such graphs, we expect our algorithm to perform worse than random walks.

Having said all that, our work demonstrates that in various real and synthetic social networks, significant query complexity savings can be obtained using suitable structural decompositions and careful preprocessing.
We stress that we do not believe that our algorithm should replace RW-based approaches. Indeed, these algorithms are the gold standard for node sampling and provide an exceptionally versatile primitive in network analysis. Rather,  
one can think of our algorithms as useful alternatives when multiple node samples are required and where the networks at question have suitable community structure.
It would be very interesting to further explore approaches that combine the query-efficiency of our methods with the flexibility of random walks.

%% file: 020-appendix.tex
\section{Appendices}
\label{sec:supp}

\subsection{Remaining procedures of \SL} 
\label{sec:preprocess_missing}

We provide in Figure~\ref{fig:preprocess} the pseudocode of two structural analysis procedures in \SL: \CompRSZ, for computing an estimate for the baseline reachability score; and \CompLTsize, for estimating the size of $L_{\geq 2}$. 


\preprocessing

\subsection{Description of \SLP}
\label{sec:pseudo}

\newcommand{\SLtwoP}{\hyperref[proc:sample-ltwoplus]{\color{black} \sc Reach-${L_{\geq 2}}$}}
\newcommand{\SLtwoALGP}{
	\alignInputOutput{
		\textbf{\textsc{Reach-$\boldsymbol{L_{\geq 2}}$}} \label{proc:sample-ltwoplus}\\
		\small
		\textsc{Input:} The data structure $\mathcal{D^+}$. \\
		\textsc{Output:} A node in $L_{\geq 2}$, and its component and reachability score.
}
\smallskip
		\begin{compactenum}
		    \smallskip
			\item  While true:
		\begin{compactenum}
            \item Use $\mathcal{D}^+$ to sample a node $u \in L_1$ according to the distribution prescribed by $\mathcal{D}^+$.
            \item Query $u$, compute $N(u)\cap L_0$,  and pick a uniform neighbor $w$ in $N(u)\setminus L_0 = N(u) \cap (L_1 \cup L_2)$. 
            \item If $w \in L_2$, fix it and exit loop. Otherwise, repeat loop.
        \end{compactenum}
        	\item Perform a BFS of $L_{\geq 2}$ to find the component $C$ containing $w$. Invoke \CompRS\ to compute its reachability $rs(C)$. 
			\item \textbf{Return} $w$,$C$, and the reachability 
			$rs(C)$.
			\end{compactenum}
}

\newcommand{\GenLZP}{\hyperref[proc:gen-L0-plus]{{ \color{black} \sc Generate-$ {L_0}$}}}
\newcommand{\ALGGenLZP}{
    \alignInputOutput{
    \textbf{\textsc{Generate-$\boldsymbol{L_0}$}} 
    \label{proc:gen-L0-plus}\\
        \small
    \textsc{Input:} arbitrary node $v_0$, number $\ell_0$. \\
    \textsc{Output:} $L_0$ of size $\ell_0$, $L_1$, $\mathcal{D}^+$.
}
    \begin{compactenum}
    
        \item \textbf{Query} $v_0$ and let $L_0=\{v_0\}$,  $L_1=N(v_0)$.
        \item Repeat $\ell_0 - 1$ times: pick $u \in L_1$ with maximum degree (break ties randomly); \textbf{query} $u$; remove $u$ from $L_1$; add $u$ to $L_0$; and add $N(u)\setminus L_0$ to $L_1$. 
        \item Create a data structure $\mathcal{D}^+$ that allows to sample a node in $L_1$ with probability proportional to its number of neighbors in  $L_1\cup L_2$.
    \end{compactenum}
}

\begin{figure}[htb!]
\begin{center}
\framebox{
\begin{minipage}{.9\columnwidth}
     \ALGGenLZP
\end{minipage}
}
\framebox{
\begin{minipage}{.9\columnwidth}
     \CompRSALGP
\end{minipage}
}
\framebox{
\begin{minipage}{.9\columnwidth}
     \SLtwoALGP
\end{minipage}
}
\end{center}
\caption{Pseudo-code for \SLP.}
\label{fig:pseudo_SLP}
\end{figure}

In Figure \ref{fig:pseudo_SLP} we present the psuedocode of the procedures in \SLP whose implementation differs from that in \SL. There are three such procedures: for generating $L_0$, reaching $L_{\geq 2}$, and computing the reachability of a node. All other procedures are identical to \SL. 
See also Figure~\ref{fig:reach-SLP} for a visualization of the sampling phase in \SLP; compare this to the analogous Figure \ref{fig:reach-SL} for \SL.
\input{022-reach-SL+}

\SLP takes advantage of the stronger query model (which also reveals the degrees of the neighbors) in two ways. First, in the $L_0$-generation phase, at each round we pick a neighbor with absolute maximum degree, rather than a ``perceived'' maximum degree as in \SL. This results in \SLP generally adding higher-degree nodes to $L_0$ compared to \SL.

The second modification is during the sampling phase, and involves the reachability computation. We utilize the stronger query model of \SLP to reach $L_{\geq 2}$ in ways that reduces the query-cost of rejection. Here, we have a way to guarantee that the random edge entering $L_2$ in our reaching attempt is uniform among all edges between $L_{1}$ and $L_2$; in \SL, we do not have such a guarantee. Thus, the reachability of a component is simply the number of edges entering it from $L_1$ divided by the component size. This makes the reachability both easier to compute (requires less queries) and more evenly distributed than in \SL.

\input{021-missing_proofs_sec3}

\balance
\subsection{Interval Length for Mixing}
\label{sec:interval_length}
In the query complexity evaluation for random walks, Section \ref{subsec:comparison_exp}, we mention that the walks are sampled every fixed interval, where the interval length should allow for mixing. Here we detail how the interval length is judiciously chosen. 

Consider a collection of $k$ random walks $W_1, \ldots, W_k$ with the same starting point, $v$. How can we determine how much steps of a random walk are required to mix? One natural way to do so is by considering the $t$'th nodes in each random walk, $W_1(t), \ldots, W_k(t)$, for different choices of $t$. We would like to set the interval length to the smallest $t$ for which $W_1(t), \ldots, W_k(t)$ are sufficiently uniform. 
This is a standard statistical estimation task. The most standard approach to solving it is by calculating the \emph{empirical distance} of the observed nodes $W_1(t), \ldots, W_k(t)$ to the uniform distribution over all nodes. 
For each of the networks we considered except for SinaWeibo, we determined the interval length as follows: we set the number of random walks to $k = n$, where $n$ is the number of nodes in the network. One can show that for random variables $X_1, \ldots, X_n \in [n]$ that are generated uniformly at random, the empirical distance to uniformity is concentrated around $1/e$. We thus picked $t$ to be the smallest for which the empirical distance to uniformity of $W_1(t), \ldots, W_n(t)$ is no more than some small parameter $\zeta$ (on real-valued networks we picked $\zeta = 0.01$, and for forest fire $\zeta = 0.03$) away from the ideal $1/e$.

To make for a fair comparison for our algorithm, we calculated what choice of the proximity parameter $\eps$ yields the same empirical distance of $\zeta$ from $1/e$. To do so we ran a grid search over various small values of $\eps$, where for each $\eps$ we averaged over 10 experiments of computing the empirical distance  from uniformity of a set of $n$ outputs of our algorithm  (with values of the other parameters, $\ell_0, s_1, s_{\geq 2}$, as mentioned above).

For SinaWeibo, the empirical distance based evaluation is computationally infeasible, since it essentially requires $k = \tilde{\Omega}(n)$. Instead we used a more efficient estimate, based on the number of collisions in $W_1(t), \ldots, W_k(t)$.  Distinguishing the uniform distribution from one that is $\eps$-far in variation distance requires only $O(\sqrt{n} / \eps^2)$ different walks~\cite{goldreich2011testing}. Here we chose $t$ to be the smallest for which the number of collisions among the different walks is less than three times the standard deviation of the number of collisions expected from the uniform distribution.

\subsection{Source Code}
\label{sec:source}
The source code for our algorithm will shortly be available at:
\begin{center}
\texttt{\href{https://github.com/omribene/sampling-nodes}{https://github.com/omribene/sampling-nodes}}
\end{center}
For reference, we also provide the source code for the random walk algorithms to which we compared our method. 
See the \texttt{README.md} file for usage instructions.

\subsection{System Specifications and Running Time}
\label{sec:specs}
We ran all experiments on a 20-core CPU configuration: Intel(R) Core(TM) i9-9820X CPU @ 3.30GHz; RAM: 128GB. The most time-consuming experiments were those involving the comparison with random walks, specifically the experiments determining the correct interval size for the random walks. In SinaWeibo, our largest graph, the running time for generating 1M random walks of proper length from one starting node was about three days. For all experiments involving our algorithm, the running time was at most a few hours (and usually up to a few minutes for the smaller networks).

%% file: 022-reach-SL+.tex
\newcommand{\drawLayersPlus}[4]{

\node[scale=3.5, #4] at (-8, -7.2) {$L_0$};
\node[scale=3.5, #4] at (-4, -7.2) {$L_1$};
\node[scale=3.5, #4] at (-0.4, -7.2) {$L_{\geq 2}$};

\fill[#1!25]  (-8,0) circle  [x radius=1.5cm, y radius=3cm] node{};
\draw[thin, color=#1]  (-8,0) circle  [x radius=1.5cm, y radius=3cm] node{};

\fill[#2!25] (-4,0) circle  [x radius=1.5cm, y radius=5cm] node{};
\draw[thin, color=#2]  (-4,0) circle  [x radius=1.5cm, y radius=5cm] node{};


\fill[#4!25]  (0.5,5) circle  [x radius=2.3cm, y radius=1.8cm] node{};
\fill[#4!25]  (0,1) circle  [x radius=1.8cm, y radius=1.5cm] node{};
\fill[#4!25]  (-1.2,-1.2) circle  [x radius=0.6cm, y radius=0.7cm] node{};
\fill[#3!25]  (-.3,-4) circle  [x radius=1.5cm, y radius=2cm] node{}; 
\draw[thin, #4]  (0.5,5) circle  [x radius=2.3cm, y radius=1.8cm] node{};
\draw[thin, #4]  (0,1) circle  [x radius=1.8cm, y radius=1.5cm] node{};
\draw[thin, #4]  (-1.2,-1.2) circle  [x radius=0.6cm, y radius=0.7cm] node{};
\draw[thin, #3]  (-.3,-4) circle  [x radius=1.5cm, y radius=2cm] node{};

\tikzstyle{client}=[draw=none,circle,minimum size=3ex, inner sep=0, fill=#1,]
\node[client] (L0-1) at (-8,2) {};
\node[client] (L0-2) at (-8,0) {};
\node[client] (L0-3) at (-8,-2) {};

\tikzstyle{client}=[draw=none,circle,minimum size=3ex, inner sep=0, fill=#4,]    
\node[client] (L1-1) at (-4,3.5) {};
\node[client] (L1-2) at (-4,1.75) {};
\node[client] (L1-3) at (-4,0) {};
\node[client] (L1-4) at (-4,-1.75) {};
\node[client] (L1-5) at (-4,-3.5) {};


\tikzstyle{client}=[draw=none,circle,minimum size=3ex, inner sep=0, fill=#3,]

\node[client] (L24-1) at (-1,-4) {};
\node[client] (L24-2) at (0,-3) {};
\node[client] (L24-3) at (-1,-5) {};
\node[client] (L24-4) at (0,-4.7) {};

\tikzstyle{client}=[draw=none,circle,minimum size=3ex, inner sep=0, fill=#4,]

\node[client] (L22-1) at (-1,1) {};
\node[client] (L22-2) at (0,2) {};
\node[client] (L22-3) at (0,0.5) {};
\node[client] (L22-4) at (1,1.3) {};

\node[client] (L21-1) at (-1,5.5) {};
\node[client] (L21-2) at (-1,4.5) {};
\node[client] (L21-3) at (0,6.2) {};
\node[client] (L21-4) at (0,4) {};
\node[client] (L21-5) at (1,6.5) {};
\node[client] (L21-6) at (1,4.5) {};
\node[client] (L21-7) at (1,3.5) {};
\node[client] (L21-8) at (2,5) {};

\node[client] (L23-1) at (-1, -1.1) {};

\tikzstyle{edgeStyle}=[#1!40, dashed, thin]
\draw[edgeStyle] (L0-1) to [bend right] (L0-2);
\draw[edgeStyle] (L0-1) to [bend right] (L0-3);

\tikzstyle{edgeStyle}=[#1!70, thin]
\draw[edgeStyle] (L0-1) -- (L1-1);
\draw[edgeStyle] (L0-1) -- (L1-2);
\draw[edgeStyle] (L0-2) -- (L1-2);
\draw[edgeStyle] (L0-2) -- (L1-3);
\draw[edgeStyle] (L0-3) -- (L1-3);
\draw[edgeStyle] (L0-3) -- (L1-4);
\draw[edgeStyle] (L0-3) -- (L1-5);

\tikzstyle{edgeStyle}=[#4!70, dashed, thin]
\draw[edgeStyle] (L1-1) to [bend right] (L1-2);
\draw[edgeStyle] (L1-1) to [bend right] (L1-3);
\draw[edgeStyle] (L1-4) to [bend right] (L1-5);

\tikzstyle{edgeStyle}=[#4, thin]
\draw[edgeStyle] (L1-1) -- (L21-1);
\draw[edgeStyle] (L1-2) -- (L21-2);
\draw[edgeStyle] (L1-3) -- (L21-2);
\draw[edgeStyle] (L1-3) -- (L22-1);
\draw[edgeStyle] (L1-3) -- (L23-1);
\draw[edgeStyle] (L1-4) -- (L23-1);
\draw[edgeStyle] (L1-4) -- (L24-1);
\draw[edgeStyle] (L1-5) -- (L24-3);

\tikzstyle{edgeStyle}=[#4!70, dashed, thin]
\draw[edgeStyle] (L21-2) -- (L22-1);
\draw[edgeStyle] (L23-1) -- (L24-1);

\tikzstyle{edgeStyle}=[#4, thin]
\draw[edgeStyle] (L21-1) -- (L21-3);
\draw[edgeStyle] (L21-3) -- (L21-5);
\draw[edgeStyle] (L21-5) -- (L21-8);
\draw[edgeStyle] (L21-2) -- (L21-4);
\draw[edgeStyle] (L21-4) -- (L21-6);
\draw[edgeStyle] (L21-4) -- (L21-7);
\draw[edgeStyle] (L21-4) -- (L21-3);

\tikzstyle{edgeStyle}=[#4, thin]
\draw[edgeStyle] (L22-1) -- (L22-2);
\draw[edgeStyle] (L22-1) -- (L22-3);
\draw[edgeStyle] (L22-3) -- (L22-4);

\tikzstyle{edgeStyle}=[#3, thin]
\draw[edgeStyle] (L24-1) -- (L24-2);
\draw[edgeStyle] (L24-2) -- (L24-3);
\draw[edgeStyle] (L24-1) -- (L24-2);
\draw[edgeStyle] (L24-2) -- (L24-4);

\draw[thin, color=#1]  (-8,0) circle  [x radius=1.5cm, y radius=3cm] node{};
\draw[thin, color=#2]  (-4,0) circle  [x radius=1.5cm, y radius=5cm] node{};
\draw[thin, #4]  (0.5,5) circle  [x radius=2.3cm, y radius=1.8cm] node{};
\draw[thin, #4]  (0,1) circle  [x radius=1.8cm, y radius=1.5cm] node{};
\draw[thin, #4]  (-1.2,-1.2) circle  [x radius=0.6cm, y radius=0.7cm] node{};
\draw[thin, #3]  (-.3,-4) circle  [x radius=1.5cm, y radius=2cm] node{};

}

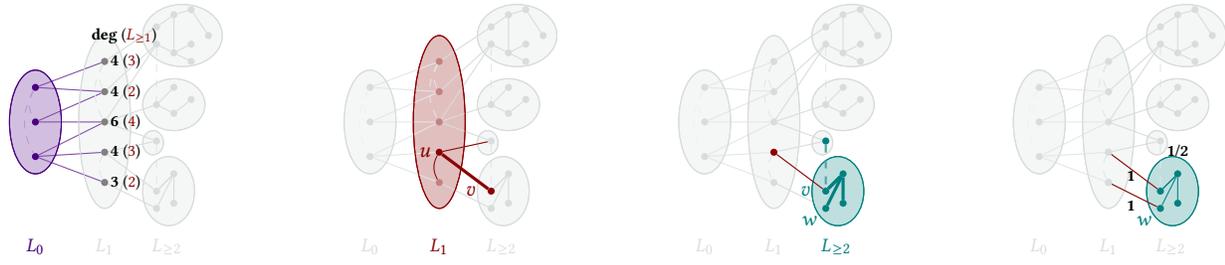
\begin{figure*}[ht]
\centering

    \hfill
  \begin{minipage}[b]{0.22\textwidth}

    \begin{tikzpicture}
    \begin{scope}[scale=0.23, transform shape]
    \drawLayersPlus{indigo}{lightgray}{lightgray}{lightgray};
    
    \tikzstyle{client}=[draw=none,circle,minimum size=3ex, inner sep=0, fill=gray,]
\node[right, scale=3] at (-5.1,5) {\textbf{deg} (\textcolor{darkred}{$L_{\geq 1}$})};
\node[client] at (-4,3.5) {};
\node[right, scale=3] at (-4,3.5)  {{\textcolor{black}{\textbf{4}} (\textcolor{darkred}{3})}};
\node[client] at (-4,1.75) {};
\node[right, scale=3] at (-4,1.75) 
{{\textcolor{black}{\textbf{4}}
(\textcolor{darkred}{2})}};
\node[client] at (-4,0) {};
\node[right, scale=3] at (-4,0) {{\textcolor{black}{\textbf{6}} (\textcolor{darkred}{4})}};
\node[client] at (-4,-1.75) {};
\node[right, scale=3] at (-4,-1.75) {{\textcolor{black}{\textbf{4}} (\textcolor{darkred}{3})}};

\node[client] at (-4,-3.5) {};
\node[right, scale=3] at (-4,-3.5) {{\textcolor{black}{\textbf{3}} (\textcolor{darkred}{2})}};
    \node[scale=3.5, indigo] at (-8, -7.2) {$L_0$};
    \end{scope}
    \end{tikzpicture}    
  \end{minipage}
  \hfill
  \begin{minipage}[b]{0.22\textwidth}

    \begin{tikzpicture}
    \begin{scope}[scale=0.23, transform shape]
    \drawLayersPlus{lightgray}{darkred}{lightgray}{lightgray};
    \draw[darkred, very thick] (L1-4) -- (L24-1) node{};
    \node[client, fill=darkred!50] at (L1-1) {};
    \node[client, fill=darkred!50] at (L1-2) {};
    \node[client, fill=darkred!50] at (L1-3) {};
    \node[client, fill=darkred] at (L1-4) {};
    \node[client, fill=darkred!50] at (L1-5) {};
    \tikzstyle{edgeStyle}=[darkred, thin];
    \draw[edgeStyle] (L1-4) -- (L23-1); 
    \draw[edgeStyle] (L1-4) to [bend right] (L1-5);
    
    \tikzstyle{client}=[draw=none,circle,minimum size=3ex, inner sep=0]
    \node[client, fill=darkred] at (L24-1) {};
    \node[scale=4, left, darkred] at (L1-4) {{$u$}};
    \node[scale=4, left=0.4, darkred] at (L24-1) {{$v$}};
    \node[scale=3.5, darkred] at (-4, -7.2) {$L_1$};
    \end{scope}
    \end{tikzpicture}

  \end{minipage}
  \hfill
  \begin{minipage}[b]{0.22\textwidth}

    \begin{tikzpicture}
    \begin{scope}[scale=0.23, transform shape]
    \drawLayersPlus{lightgray}{lightgray}{teal}{lightgray};
    \tikzstyle{client}=[draw=none,circle,minimum size=3ex, inner sep=0]
`   \node[client, fill=darkred] at (L1-4) {};
    \node[scale=4, left=0.4, teal] at (L24-1) {{$v$}};
    \node[scale=4, below left, teal] at (L24-3) {{$w$}};
    \node[scale=3.5, teal] at (-0.4, -7.2) {$L_{\geq 2}$};
    \tikzstyle{edgeStyle}=[darkred, thin];
    \draw[edgeStyle] (L1-4) -- (L24-1);
    \node[client, fill=teal] (L23-1) at (-1, -1.1) {};
    
    \tikzstyle{edgeStyle}=[teal, very thick]
    \draw[edgeStyle] (L24-1) -- (L24-2);
\draw[edgeStyle] (L24-2) -- (L24-3);
\draw[edgeStyle] (L24-1) -- (L24-2);
\draw[edgeStyle] (L24-2) -- (L24-4);
    
    \tikzstyle{edgeStyle}=[teal!70, dashed, thick]
\draw[edgeStyle] (L23-1) -- (L24-1);

    \end{scope}
    \end{tikzpicture}
  \end{minipage}
  \hfill
  \begin{minipage}[b]{0.22\textwidth}

    \begin{tikzpicture}
    \begin{scope}[scale=0.23, transform shape]
    \drawLayersPlus{lightgray}{lightgray}{teal}{lightgray};
    \tikzstyle{client}=[draw=none,circle,minimum size=3ex, inner sep=0]
    \node[client, fill=lightgray] at (L1-4) {};
    \node[client, fill=lightgray] at (L1-5) {};
    \node[scale=4, below left, teal] at (L24-3) {{$w$}};
    \tikzstyle{edgeStyle}=[darkred, thin];
    \draw[edgeStyle] (L1-4) -- (L24-1);
    \draw[edgeStyle] (L1-5) -- (L24-3);
    \node[scale=3, below right=1] at (L1-4) {\textbf{1}};
    \node[scale=3, below right=1] at (L1-5) {\textbf{1}};
    \node[scale=3,above=0.7] at (L24-2) {\textbf{1/2}};

\tikzstyle{client}=[draw=none,circle,minimum size=3ex, inner sep=0, fill=darkred,]    
    \end{scope}
    \end{tikzpicture}
  \end{minipage}
  \hfill
  \caption{Sampling a node from $L_{\geq 2}$ in \SLP. Here, the degrees of nodes in $L_1$ are known after the preprocessing is finished. Thus, for each node $u \in L_1$, we can calculate how many neighbors it has in $L_{\geq 1} = L_1 \cup L_{\geq 2}$ (left). We pick a node $u \in L_1$ with probability proportional to the number of such neighbors, and pick a random $L_{\geq 1}$-neighbor $v$ of $u$ (middle left). If $v \in L_1$, we break and return to the first step. Otherwise (middle right), $v \in L_2$ and we explore its component $C$ as in \SL, picking a random node $w \in C$. Finally (right), we compute reachability for $w$ and run a rejection step accordingly.}
  \label{fig:reach-SLP}
\end{figure*}

%% file: 021-missing_proofs_sec3.tex
\subsection{Missing proofs from Section~\ref{sec:lower_bound}
} \label{sec:missing_lb}
In this section  we prove Theorem~\ref{thm:lb} regarding a lower bound on random walks.

\begin{theorem}[Theorem~\ref{thm:lb}, restated]
For any $n$ and $\log n\ll t \ll n^{\Theta(1)}$, there exists a graph $G$ on $n$ vertices 
with mixing time $\tmix=\Theta(t)$, that satisfies the following: for any $N \leq n^{\Theta(1)}
$, any sampling algorithm based on uniform random walks that outputs a (nearly) uniform collection of $N$ nodes must perform $\Omega(N\cdot \tmix)$ queries.
\end{theorem}
\begin{proof}[Proof Sketch]
Consider a graph $G$ defined as follows. $G$ consists of a ``core'' expander subgraph $C$ over $n/2$  nodes and a subgraph $W$ over $n/2$ nodes that consists of 
$\ell=n/2t$ small subgraphs, $W_1, \ldots, W_{\ell}$, to be described soon.
We connect some node $w_i$ in each $W_i$ by a ``bridge'' to a randomly chosen node $c_i\in C$, where the distribution over the nodes from $C$ is proportional to their degree in $C$. The random choice is independent between different bridges.
The $W_i$ components should satisfy the following: (i) $|W_i|=t$; (ii) the nodes of $W_i$ have constant average degree; (iii) the mixing time of each $W_i$ is $O(t)$; and (iv) a walk starting in $w_i$ is likely (has constant probability) to mix in $W_i$ before it exits $W_i$ through the bridge. For example, taking $W_i$ as a star or an expander would suffice.

Since  $|W|=n/2$, a constant fraction of the node samples should be from $W$, and since every pair $W_i$ and $W_j$ is disconnected, this necessarily means that the random walk should visit $\Theta(N)$ distinct components $W_i$, and thus cross $\Theta(N)$ distinct bridges.\footnote{By the birthday paradox, when sampling $N = o(\sqrt{n/t})$ uniform and (nearly) independent nodes, the probability that two of these $N$ nodes will come from the same component  is negligible. This implies that $\Omega(N)$ different components must be visited.}
 Furthermore, as the probability of leaving $C$ is bounded by $O(\frac{n/t}{n})=O(1/t)$, a walk starting in some node in $C$ is not likely to leave $C$ unless $\Omega(t)$ steps are taken. Together with the fact that the walk should cross a total of $\Omega(N)$ bridges, this implies an $\Omega(N\cdot t)$ lower bound on the number of required steps. It is easy to choose $C$ in such a way that the number of required queries is of the same order.

It remains to prove  that $\tmix(G)=\Theta(t)$. 
Due to space limitations, we only give some intuition, and the formal proof  can be achieved by bounding the hitting times of large sets of $G$, as these measures are  equivalent up to a constant~\cite{peres2015mixing}.
Since $\Omega(t)$ steps are required to cross a bridge from $C$ to some $W_i$, $\tmix = \Omega(t)$. 
The other direction is a bit less direct, but intuitively,  since in expectation leaving $C$ takes $\Omega(t)$ steps and the mixing time of $C$ is $O(\log n)$, the walk will be uniform over the edges of $C$ before leaving to $W$. 
Therefore, after  $\Omega(t)$ steps the walk  has constant probability of reaching $W$, and each $W_i$ is equally likely to be the one reached. Once in $W_i$ for some $i$, by condition (iv) in the first paragraph, the walk has constant probability to mix before leaving the component. Moreover, after $\Theta(t)$ steps in $W_i$ the walk has constant probability to return to $C$, where it will again typically mix before leaving. 
Hence, after $\Theta(t)$ steps the walk is mixed over the entire graph $G$. 
A similar argument holds for the case that the walk starts at some vertex $w\in W$.
Therefore, $\Theta(t)$ steps are  necessary and sufficient for the walk to mix. 
\end{proof}

The above proof shows that a rather generic family of graphs consisting of an expander component and many isolated communities satisfies this bound.
We note that an even more general but slightly weaker lower bound, that applies to \emph{all} networks with a substantial number of small weakly connected components, can be proved using the notion of conductance \cite{Sinclair1994}.

\subsection{Missing proofs from Section~\ref{sec:ub}}\label{sec:size_estimation}

In this section we provide missing proofs from Section~\ref{sec:ub}.
 We start by claiming that 
 our algorithm accurately estimates the size of $L_{\geq 2}$, and   then continue to prove that our algorithm induces a close to uniform distribution.
 
We make use of the following notation. Denote by $L_i$ the collection of all nodes of distance exactly $i$ from $L_0$ (whereas $L_0$ is formally defined as the output of the {\GenLZ} procedure). 
For a node $v \in L_i$, for $i=1,2$, we denote by $\dm(v)$ and $\dpp(v)$ its number of neighbors in $L_{i-1}$ and $L_{i+1}$, respectively. For nodes in $L_{\geq 2} \setminus L_2$, we set $\dm(v)  = \dpp(v) = 0$.

For a set $R$, let $\dpp(R)=\sum_{v \in R} \dpp(v)$ and $\dm(R)=\sum_{v \in R} \dm(v)$. Also, let $\dpp_{avg}(R) = \dpp(R)/|R|$, and $\dpp_{\max}(R) = \max_{v \in R}\{ \dpp(v)\}$, and analogously for   $\dm_{\max}(R)$ and $\dm_{avg}(R)$. 
Furthermore, we denote by $rs(v)$ the reachability score of a node in $L_{\geq 2}$ as computed by the algorithm. For a set $S \subseteq L_{\geq 2}$, we define $rs(S) = \sum_{v \in S} rs(v)$. The average reachability over the set is $rs_{avg}(S) = rs(S) / |S|$. Finally, let $rs_{min}(L_{\geq 2}) = \min_{v \in L_{\geq 2}} rs(v)$ and $c_{rs}=rs_{avg}(L_{\geq 2})/rs_{min}(L_{\geq 2})$.

\begin{theorem}
Let $\bar \ell_{\geq 2}$ be the output of {\CompLTsize} for the input parameters $$s_1=
\Theta\left(\frac{\dpp_{max}(L_1)}{\eps^2\cdot \dpp_{avg}(L_1)}\right), \ \ \  s_{\geq 2}=\Theta\left(\frac{c_{rs}  \cdot d^-_{max}(L_{\geq 2})}{\eps^2\cdot d^-_{avg}(L_{\geq 2}) }\right).$$ Then with high constant probability, $\bar \ell_{\geq 2} \in (1\pm \eps)|L_{\geq 2}|$.
\end{theorem}
\begin{proof}
Let $S_1$ denote the set of $s_1$ (uniform and independent) samples from $L_1$. By the multiplicative Chernoff bound, it suffices to take $s_1 = O\left(\frac{d^+_{max}(L_1)}{\eps^2 d^+_{avg}(L_1)} \right)$, so that 
with high constant probability 
\begin{align}\label{eq:dplusL1}
d^+_{avg}(S_1) \in (1\pm \eps/4)d^+_{avg}(L_1).
\end{align}

Consider the samples $S_2$ from $L_{\geq 2}$. 
Let $rs(L_{\geq 2})=\sum_{v \in L_{\geq 2}}rs(v)$.
Each node in $L_{\geq 2}$ is reached in an invocation of \SLtwo\ with probability $rs(v)/\sum_{v \in L_{\geq 2}} rs(v) = rs(v)/rs(L_{\geq 2}).$ 
In the following, the expectation is over a reached node returned by \SLtwo.  
\begin{align*}
    \EX_{v}
    \left[\frac{d^-(v)}{rs(v)}\right]
    =\sum_{v \in L_{\geq 2}} \frac{rs(v)}{rs(L_{\geq 2})}\cdot \frac{d^-(v)}{rs(v)} 
    =\frac{d^-(L_{\geq 2}) }
    {rs(L_{\geq 2})}
=\frac{\dm_{avg}(L_{\geq 2})}{rs_{avg}(L_{\geq 2})}.
\end{align*}
\sloppy
By the multiplicative Chernoff bound,
$s_2=\Theta\left(\frac{c_{rs}  \cdot d^-_{max}(L_{\geq 2})}{\eps^2\cdot d^-_{avg}(L_{\geq 2}) }\right)$ is sufficient so that with high constant probability 
\begin{align}
\sum_{v \in S_2}\frac{d^-(v)}{rs(v)} \in (1 \pm \eps/4) \cdot |S_2|\cdot  \frac{d^-_{avg}(L_{\geq 2})}{rs_{avg}(L_{\geq 2})}. \label{eq:davg}
\end{align}

Using the fact that $s_2 \geq \Theta\left( \frac{rs_{avg}(L_{\geq 2})}{\eps^2\cdot rs_{0}}\right)$ and 
\begin{align*}
    \EX_{v}
    \left[\frac{1}{rs(v)}\right]=\sum_{v\in L_{\geq 2}} \frac{rs(v)}{rs(L_{\geq 2})}\cdot \frac{1}{rs(v)} =\frac{|L_{\geq 2}|}{rs(L_{\geq 2})} = \frac{1}{rs_{avg}(L_{\geq 2})},
\end{align*}
by the multiplicative Chernoff bound, 
the total reachability score that we compute in our pseudocode, $trs$, satisfies
\begin{align}
trs(S_2)=\sum_{v \in S_2}\frac{1}{rs(v)} \in (1\pm \eps/4)\cdot \frac{|S_2|}{rs_{avg}(L_{\geq 2})} \label{eq:trs}
\end{align}
with high constant probability.
By Equations~\eqref{eq:davg} and~\eqref{eq:trs}, w.h.c.p. $$\frac{1}{trs(S_2)}\cdot \sum_{v \in S_2} \frac{d^-(v)}{rs(v)}
\in (1\pm \eps/2)d^-_{avg}(L_{\geq 2}).$$ 

Finally, it holds that $|L_1|\cdot d^+(L_1) =|L_{\geq 2}|\cdot d^-(L_{\geq 2})$. Therefore,  plugging into this equation 
$\bar d^+_{avg}(L_1)$
and
$\bar d^-_{avg}(L_{\geq 2})$, yields a $(1\pm\eps)$ approximation of $|L_{\geq 2}|$ with high constant probability. 
\end{proof}

Next we prove that the distribution induced by our algorithm's samples is close to uniform.
\begin{theorem}[Theorem~\ref{thm:convergence_uniformity}, formally stated]
If our size estimation for $L_{\geq 2}$, $\bar \ell_{\geq 2}$, is within $(1 \pm \delta)|L_{\geq 2}|$, and if the baseline reachability $rs_0$ used in our algorithm is the  $\eps^{\textrm{th}}$ percentile in the reachability distribution, then the induced distribution on nodes by the procedure \Sample\ is $(\eps+\delta + o(\eps, \delta))$-close to uniform in total variation distance, where the lower order term is $o(\eps, \delta) = O((\eps + \delta)^2)$. Furthermore, the probability of any given node to be sampled is at most $(1+\eps +\delta + o(\eps, \delta)) / n$.
\end{theorem}
\vspace{-0.2cm}
\begin{proof}
By definition of the total variation distance, it suffices to prove the last statement of the theorem, regarding the probability bound for any given element; the first statement 
follows.

Let $\bar n=|L_0|+|L_1|+\bar \ell_{\geq 2}$ be the estimated total size of the network. If  $\bar \ell_{\geq 2}$ is in $(1\pm \delta)|L_{\geq 2}|$ then $\bar n \in (1\pm \delta)n$. Condition on this event.
For each node in $L_0$, its probability to be returned by \Sample\ is $\frac{|L_0|}{\bar n}\cdot \frac{1}{|L_0|} = \frac{1}{\bar n} \leq \frac{1 + \delta + O(\delta^2)}{n}$. Similarly this holds for nodes in $L_1$. Thus, it remains to prove the upper bound for nodes in $L_{\geq 2}$. 

For every $v\in L_{\geq 2}$, define by $rs(v)$ the unique value for which $v$'s component is reached  by \SLtwo\ with probability $\frac{rs(v)}{d^+(L_0)}$, where $d^+(L_0)$ is the number of edges from $L_0$ to $L_1$.
Let $A(L_{\geq 2})$ denote the set of nodes in $L_{\geq 2}$
for which $rs(v) \geq rs_0$.
In a single invocation of the while loop of \SLtwo, 
each node $v\in A(L_{\geq 2})$ is returned with probability $\frac{rs(v)}{d^+(L_0)} \cdot \frac{1}{rs(v)} = \frac{1}{d^+(L_0)}$. 
For a node $v\in L_{\geq 2}\setminus A(L_{\geq 2})$, its probability of being returned is between zero and $1/d^+(L_0)$.
Therefore, \SLtwo\xspace is uniform over  $A(L_{\geq 2})$ (but is biased against nodes with extremely small reachability, i.e., those in $L_{\geq 2} \setminus A(L_{\geq 2})$). 

It follows that nodes in $A(L_{\geq 2})$ are returned with probability at least $\frac{\bar \ell_{\geq 2}}{\bar n}\cdot \frac{1}{|A(L_{\geq 2})|}$ and at most $\frac{\bar \ell_{\geq 2}}{\bar n}\cdot \frac{1}{|L_{\geq 2}|}$. 
By the assumption that for $1-\eps$ of the nodes in $L_{\geq 2}$,  $rs(v)\geq rs_0$, it holds that $|A(L_{\geq 2})|\in [(1-\eps),1] |L_{\geq 2}|$. Together with the assumption that $\bar \ell_{\geq 2} \in (1\pm \delta) |L_{\geq 2}|$,  
we get that each node in $A(L_{\geq 2})$ is returned with probability at most $(1-\eps)^{-1} (1-\delta)^{-1} /n \leq (1+\eps+\delta+O((\eps+\delta)^2) / n$. 
The other nodes in $L_{\geq 2}$ are returned with a smaller probability. 
\end{proof}